\documentclass[a4paper,12pt]{article}
\usepackage[a4paper,margin=1.00in]{geometry}
\RequirePackage{amsthm,amsmath,amsfonts, amssymb}
\usepackage[dvipsnames]{xcolor} 
\usepackage{bbm} 
\newtheorem{theorem}{Theorem}
\newtheorem*{theorem*}{Theorem}
\DeclareMathOperator{\Var}{Var}
\DeclareMathOperator{\argmax}{argmax}
\usepackage{graphicx}
\usepackage[authoryear]{natbib} 
\usepackage{subcaption}
\usepackage[english]{babel}
\usepackage{hyperref}

\let\oldabstract\abstract
\let\oldendabstract\endabstract
\makeatletter
\renewenvironment{abstract}
{%
               {\list{}{\addtolength{\leftmargin}{-2em} 
                        \listparindent 1.5em%
                        \itemindent    \listparindent%
                        \rightmargin   \leftmargin%
                        \parsep        \z@ \@plus\p@}%
                \item\relax}%
               {\endlist}%
\oldabstract}
{\oldendabstract}
\makeatother

\title{Revisiting the effect of search frictions on market concentration}
\author{Jules Depersin and Bérengère Patault\thanks{Jules Depersin: University of Amsterdam, j.d.p.depersin@uva.nl ; Berengere Patault: University of Amsterdam,  b.patault@uva.nl. We thank Rémi Avignon, Simon Bunel, Pieter Gautier, Etienne Guigue, Albert Jan Hummel, Isabelle Mejean, Horng Chern Wong for their remarks.}}
\date{\today\\ Working paper}

\begin{document}

\maketitle
\begin{abstract}
Search frictions can impede the formation of optimal matches between consumer and supplier, or employee and employer, and lead to inefficiencies. This paper revisits the effect of search frictions on the firm size distribution when challenging two common but strong assumptions: that all agents share the same ranking of firms, and that agents meet all firms, whether small or large, at the same rate. We build a random search model in which we relax those two assumptions and show that the \textit{intensity} of search frictions has a non-monotonic effect on market concentration.  
An increase in friction intensity increases market concentration up to a certain threshold of frictions, that depends on the slope of the meeting rate with respect to firm’s size. We leverage unique French customs data to estimate this slope. First, we find that in a range of plausible scenarios, search frictions  intensity increases market concentration.
Second,  we show that slopes have increased over time, which unambiguously increases market concentration in our model. Overall, we shed light on the importance of the \textit{structure} of frictions, rather than their \textit{intensity}, to understand market concentration.




\vspace{1.5cm}

\end{abstract}

\clearpage
\section{Introduction}





Market concentration has increased in the last decades in the US \citep{akcigit2021ten,autor2020fall,decker2016declining}. Yet, economists are still far from reaching a consensus on the causes of this increase. Two groups of explanations prevail: differentiating them is crucial as they lead to vastly different policy implications. 
On the one hand, the rise in  concentration might be due to a rise in barriers to entry: in this case, market concentration is bad news, as it increases market power, and diminishes innovation and efficiency \citep{de2021quantifying,gutierrez2019failure}. On the other hand, the rise in concentration might be due to an increase in market competition, which enables the most-efficient firms to grow faster: in this case, market concentration is good news \citep{autor2020fall}. For instance, the recent decrease in search frictions, with the advent of digital platforms among others, has been advanced as one explanation for increased market competition and thus increased market concentration \citep{albrecht2022vertical,bai2020search}.
The common view is that by preventing  agents from being matched with the most-efficient firm immediately, search frictions allow  the least-efficient firms to survive and thus decrease both market concentration and market competition \citep{eaton2019firm,lenoir2022search}. 
In this paper, we revisit this theoretical link between search frictions and market concentration. Our main result is twofold. First, we show that the \textit{intensity} of search frictions has a non-monotonic effect on market concentration. In most realistic settings, market concentration increases with the \textit{intensity} of search frictions and thus increases while market competition decreases. 
Second, we shed light on the importance of changing \textit{structures} of search frictions over time rather than their \textit{intensity} as drivers of concentration.  \\






In a world of homogeneous preferences, all the agents - which could be either consumers or workers - should ideally be matched with the firm that is unanimously preferred - often referred to as the most efficient. All other firms should optimally have a null size: frictions can only decrease market concentration. 
In this scenario, an increase in market concentration is always efficient, as it indicates that a larger fraction of agents have found the most efficient firms and are matched with them. 
Yet, there is a substantial evidence that agents have heterogeneous preferences: workers with different skills rank jobs differently and climb different job ladders \citep{lindenlaub2016multidimensional,moscarini2001excess,papageorgiou2014learning}, consumers have different valuations of product characteristics (see the IO literature on demand estimation, \cite{berry1995automobile}).\footnote{ In this paper, we are agnostic about the reasons for heterogeneous agents' preferences. These preferences can arise from a variety of sources. One of them is that agents may have multi-dimensional traits: for instance workers may have both a cognitive ability and a manual ability ; and employers may have different needs for cognitive versus manual skills. With such multi-dimensionality, different workers rank employers differently \citep{lindenlaub2016multidimensional}. Another source of heterogeneous preferences may be non-homothetic preferences: poor and wealthy consumers may have different utility weights on the product's quality and price. As a result, consumers do not necessarily all minimize the quality-adjusted price of products.  
} 
In the presence of heterogeneous preferences, the optimal assignment changes, and each agent wants to be matched with their own preferred firm. The efficient equilibrium has each firm matched with the masses of agents that prefer them. In this case, frictions still induce deviations between the optimal and the realized size distribution. Yet, we do know whether the realized distribution is more or less concentrated than the optimal one. Do search frictions smooth the distribution of matches or can they concentrate it further? How does a decrease in search friction affect the concentration of the market ?  \\


We build the first theoretical model that studies the effect of frictions on market concentration, allowing for heterogeneous preferences. We find that the effect of search frictions on market concentration is not monotonic in friction intensity. An increase in friction increases market concentration up to a certain threshold of frictions, that depends on the slope of the meeting rate with respect to firm's size. Note that our theory encompasses the  extreme case in which the meeting rate is uniform across firms - the most common assumption in the literature: in that case the threshold is null so that frictions always decrease market concentration. The slope at which the meeting rate increases with the firm's size is therefore crucial to characterize the effect of frictions for market concentration. We call this slope the ``structure'' of the search frictions, which has been largely overlooked by the literature.  
 We use data from the international goods market to estimate this structure empirically.  We find that in a range of plausible scenarios, search frictions increase market concentration, which casts doubt about the link between rising market concentration and decreasing friction intensity. However, using our model, we find that changes in the structure itself throughout the years, which we observe empirically, can cause market concentration while decreasing total welfare, making it a ``bad concentration''.    \\ 

To build intuition, we start with a toy model: two firms $A$ and $B$ match with many agents.  While firms have no capacity constraint and are always happy to match with any agent that they encounter, agents can be matched with at most one firm, but keep searching after their first match as they can switch from one firm to another. We assume there are two types of agents: some have an intrinsic preference for firm $A$, while others prefer firm $B$.\footnote{We are agnostic about why different agents may prefer either firm $A$ or $B$. The differences between the two companies could be due to a range of factors, such as productivity, facilities, or quality, which we do not specify. 
} We first show that when the meeting rate is the same for both firms, higher frictions advantage the least preferred firm of the market: this firm's market share is higher in equilibrium than the share of agents who prefer this firm. Then, we show what happens if a firm's meeting rate depends on the firm's market share. We show empirically that this hypothesis can be more realistic in a number of settings. In this case, an increase in friction intensity has a non-monotonic effect, which is a unique feature of our model: it first advantages the most-preferred firm of the market, then disadvantages it. The level of frictions at which frictions start to disadvantage the most-preferred firm is an increasing function of the slope at which meeting rates increase with firms size.
While frictions have been described in the literature as the force keeping small firms alive, this shows that they can on the contrary lead to the disappearance of small firms which would otherwise survive.  We therefore claim that the
way meeting rates depend on the size of firms matters and can change our representation
of how frictions distort firm sizes and of who benefits from them.  \\


We then build a more general framework with a continuum of suppliers. We develop a partial equilibrium model of the matching between atomistic agents and firms. The model is a random search model that features two-sided heterogeneity - agents and firms - and many-to-one matching: one agent can be matched with one firm at most, and the equilibrium size of a firm is determined by the number of agents to whom it is matched. In contrast with most of the literature, we assume that the agents have intrinsic preferences for certain firms, which we formalize building on \cite{lindenlaub2016multidimensional}. The optimal firm size distribution is the distribution of agents' preferences.  We show that in this more general setup, basic features of our toy-model are still present. Frictions tend to have an homogenizing (or ``deconcentrating'') effect on the firm size distribution when compared with the optimal one \emph{when the meeting rate is the same for all firms}. We show that in other cases, where the meeting rate is allowed to depend on the market share of the firm, frictions can lead to an increase of market concentration, which is consistent with the findings of the two firms model and which is a unique feature of our model. In such cases, an increase in friction increases market concentration up to a certain threshold of frictions, that depends on the slope of the meeting rate with respect to firm’s size. It is worth noting that our model is applicable to any many-to-one matching framework, and is thus readily applicable to labor markets or goods markets. \\

We model the meeting rate of a firm as a affine function of the firm's size, where the slope with respect to the firm's size is equal to a parameter $\alpha$.\footnote{One potential interpretation is that this parameter corresponds to the share of agents that meet firms proportionally to their size, while the other agents meet firms at a rate independent of their sizes. Agents may indeed use different methods to meet firms: some might be using digital platforms when other are not for instance.}  We then show that we can compute with our theoretical model an optimal $\alpha$ that maximizes the matching efficiency, defined as the sum of surpluses of all matches, for a given friction intensity and for a given distribution of agents' preferences. This  optimal $\alpha$ differs from the private optimal $\alpha$ that an individual agent would choose if given the choice, as the agent does not internalize the effect of her own match on the firm size distribution and thus on the meeting rates of other agents. This paper thus introduces a new phenomenon of concentration externality. \\


The conclusion of our theoretical study is that depending on the shape of search frictions, higher frictions may either advantage the least or the most preferred suppliers of the market. To gain further insights about what happens in the real world, we study the shape of search frictions using French data.
We use a detailed firm-to-firm export dataset, provided by French Customs, covering the universe of French firms and their exports to European Union destinations.
We confirm, regressing the flow of previously unmatched buyers coming into a firm at time $t$ on the size of this firm at time $t-1$ that the meeting rate is well approximated by a linear function of the market share of suppliers. We estimate that $\hat{\alpha} = 0.75$, which indicates that the meeting rate increases close to one-to-one with the supplier's market share. 
We provide numerical evidence that with different scenarii of agents' preferences and with such values of $\alpha$, frictions mainly have a concentrating effect: search frictions tend to advantage the most preferred firms of the market. \\  

We study the evolution of $\alpha$ over the years and find that it has increased significantly over the past 20 years: it rose from 0.63 in 1996 to 0.84 in 2016. This finding has important policy implications. Market concentration has increased in the last decades, which is seen as
worrisome by both policy-makers and economists \citep{akcigit2021ten,autor2020fall,decker2016declining} A potential explanation is a decrease in search frictions, with the advent of digital platforms among others \citep{albrecht2022vertical,bai2020search,lenoir2022search,mazet2021information}.
Our findings challenge and nuance this theory: 
we suggest that part of the rise in market concentration can be due to a change in the \emph{structure} of search frictions rather than a change in their \emph{intensity}. This new explanation has dramatically different policy implications than the standard view: while a decline in frictions magnitude is welfare-enhancing, an increase in the slope $\alpha$ is not. In the end, we show that the advent of digital platforms may lead to what we call a ``bad concentration'' that decreases efficiency. \\

This paper is organized as follows. Section \ref{litterature} reviews the literature.  Section~\ref{sec:twosup} presents a toy model with two firms, which enables the reader to gain intuition about the equilibrium effect of search frictions. Section~\ref{sec:multifirms} shows that our results on the effects of frictions generalize to an environment with a continuum of firms.  Section~\ref{sec:heterogmeetingrates} provides empirical evidence that the meeting rate of a firm does increase with this firm's market share. Section~\ref{sec:conclusion} concludes.

\section{Related literature} \label{litterature}

Our paper, through analyzing the role of search frictions for the firm size distribution - contributes to three large and distinct literatures: the Labor literature, the Trade literature and the Industrial Organization literature.  \par 
A vast, but relatively recent, body of research describes the presence of information frictions in international goods markets \citep{allen2014information,bernard2022origins,chaney2014network,dasgupta2018inattentive,eaton2019firm,fontaine2022frictions,krolikowski2021goods,lenoir2022search,ornelas2021preferential,startz2016value,steinwender2018real}. Those papers explore the magnitude of such frictions, and their effect on equilibrium outcomes, such as the equilibrium number of buyers per firm \citep{lenoir2022search,eaton2019firm,fontaine2022frictions}. \cite{lenoir2022search} explain that search frictions distort the effectiveness of resource allocation across heterogeneous producers, through allowing the least-efficient firms to survive.\footnote{This increased competition due to a decline of search frictions has also been emphasized on domestic product markets \cite{albrecht2022vertical} or credit markets \cite{mazet2021information}.} Yet, this result hinges on the assumption of a common ranking of suppliers across buyers. We contribute to this literature by including, for the first time, heterogeneous buyers' preferences for suppliers in a model of international goods markets. In this setting, search frictions may impact the effectiveness of matching through novel mechanisms. Notably, we show that, taken together with the fact that the rate at which foreign buyers meet suppliers increases with the supplier's market share, higher frictions may lead to a higher market concentration, which is a unique feature of our model. \\\par 
The Labor literature was the first to model search frictions. A vast number of papers has analyzed the matching of workers to firms, to explain the dispersion of wages as well as to study the assortativity of matching between workers and firms \citep{postel2002equilibrium}. 
Formally, our paper builds on a recent paper \citep{lindenlaub2016multidimensional}, which analyzes sorting in a framework where workers and firms have multi-dimensional characteristics. This multi-dimensionality generates sorting between workers and firms that cannot be obtained in simpler settings.\footnote{More precisely, to generate sorting, one needs either multi-dimensionality of agents' characteristics, capacity constraints on the firm side, or type-dependent peaks of the surplus function.} 
We re-interpret their model as different agents of type $x$ having heterogeneous preferences for firms of type $y$. 
We contribute to this literature by showing that, with meeting rates increasing in firms' sizes on top of heterogeneous preferences,  the effect of search frictions on firm size and on the efficiency of the matching process is significantly altered. \\

Our paper also closely relates to the existing literature on information frictions in online marketplaces \citep{bai2020search,hui2022designing,li2020buying}. \cite{bai2020search} show that digital platforms have a limited ability to help small high-quality sellers overcome the demand frictions, because of the too large market congestion. They emphasize in particular the demand reinforcement effect, through which past sales cause the arrival of future sales. This demand reinforcement effect is equivalent to our finding that the rate at which an agent can meet a firm increases with the firm's market share. We contribute to this literature by showing that this demand reinforcement effect has consequences on the effect of search frictions on the firm size distribution. In particular, an increase in this effect could explain part of the increase in market concentration in the recent decades. \\
\par
Last, our paper relates to the macroeconomic literature on the misallocation of resources in the economy, and its consequences for aggregate output and growth \citep{david2016information,hsieh2009misallocation,restuccia2017causes}. \\

\section{Two-firms model} \label{sec:twosup}

To gain insight into the effect of search frictions on the allocation of agents to firms, we present a toy model with two firms and heterogeneous agents. We deviate from the common assumption that all agents prefer the most productive firm of the market, and introduce heterogeneous agents' preferences. We study the effect of frictions on the allocation of agents to firms in this setup. In a world with homogeneous agents' preferences, search frictions always advantage the least preferred firms. We show that with heterogeneous agents' preferences this needs not be true. In particular we show that how one models search frictions is not innocuous: in some cases frictions advantage the least preferred firms in the market, and in other the most preferred firms.


\subsection{Main assumptions}

We consider two firms $A$ and $B$. Agents enter the market at rate $\mu$. When entering the market, the agents are unmatched and look for firms. They then meet firm $A$ at a rate $\lambda_A$, and firm $B$ at a rate $\lambda_B$.  The agents match initially with the first firm they meet.\footnote{Making this assumption is equivalent to assuming that all matches lead to a positive surplus. \\} After this first match, agents keep on searching for new firms, and meet the other firm, whether $A$ or $B$, at respective rate $\lambda_{A}$ or $\lambda_{B}$.
\footnote{This hypothesis is made in this section to keep the toy model as simple as possible. It is however unrealistic that unmatched and matched agents meet firms at the same rate. All results presented continue to hold when allowing meeting rate for matched agents to be the one for unmatched ones divided by a constant. Our multi-firm model in the next section uses this more realistic hypothesis. } When they meet the other firm, they can switch firms or remain with their current match, depending on their \emph{personal preferences}. We assume that firms have no capacity constraint and that all matches yield a positive surplus. Therefore firms are always happy to match with any agent that they encounter. There are two types of agents: type $a$ agents have an intrinsic preference for firm $A$, while type $b$ agents prefer firm $B$. We denote $p_a$ the proportion of agents which prefer firm $A$ while a proportion $p_b=1-p_a$ prefers firm $B$. Note that in most existing models of agent-firm matching with frictions, one assumes that all agents prefer the same firm, the one with the highest productivity, and therefore $p_a \in \{0,1\}$. \\ 

Agents exit the market at rate $\mu$, independently of whether they are matched and whom they are matched with. We denote $u(i)$ the share of unmatched agent of type $i$, $u$ the total share of unmatched agents, $u=u(a)+u(b)$ and $m=1-u$ the share of matched agents.  We define the total meeting rate $\lambda_{tot} = \lambda_A + \lambda_B$. \\

We denote $h(i,A)$ and $h(i,B)$ the number of type $i$ agents matched with respectively firm $A$ and with firm $B$. The number of agents matched with firm $A$ writes: $h(A)=h(a,A)+h(b,A)$. \\

At steady state the inflows and outflows into and out of each state compensate each other. 
For instance, the number of type $i$ agents entering the unmatched state should equate the number of type $i$ agents leaving the unmatched state. In each period, a proportion $\mu \times p_i \times dt$ of type $i$ agents enter the market and become unmatched: these are the inflows into the unmatched state. The outflows are composed of: a proportion $(\lambda_A+\lambda_B)$ of the unmatched agents of type $i$ which become matched, and a proportion $\mu$ of the unmatched agents of type $i$ which exit the market. At equilibrium, inflows into the unmatched state equal outflows, therefore:   $$(\lambda_A+\lambda_B+\mu) u(i) = \mu p_i.$$

From this equation we can compute $u(i)$ and recover $m=\lambda_{tot}/(\mu+\lambda_{tot})$.

The inflows of type $a$ agents into firm $A$ are equal to the unmatched agents of type $a$ - $u(a)$ - which meet $A$, which happens at rate $\lambda_A$, and the type $a$ agents matched initially with $B$, i.e. $h(a,B)$, which meet A. The outflows are type $a$ agents matched with $A$ which exit the market. Therefore at equilibrium:
\begin{align*}
\mu h(a,A) = h(a,B) \lambda_A + u(a) \lambda_A.
\end{align*}

The inflows of type $a$ agents into firm $B$ are the unmatched agents of type $a$ which meet $B$. The outflows are the type $a$ agents initially in $B$ which either exit the market, which happens at rate $\mu$, or meet $A$, which happens at rate $\lambda_A$:
\begin{align*}
    \mu h(a,B) + \lambda_A h(a,B) = \lambda_B u(a)
\end{align*}

Writing those equations with type $b$ agents and re-arranging gives the equilibrium number of agents matched with firm $A$ and firm $B$:
\begin{align*}
    h(A) &=  \frac{\lambda_{tot}}{\mu + \lambda_{tot}} \left[ 
    \left( \frac{\lambda_A}{\mu+\lambda_B} \times \frac{\mu}{\lambda_{tot}} \right) 
    + p_a 
    \left( \frac{\lambda_A \lambda_B (2\mu+\lambda_{tot})}{\lambda_{tot}  (\mu+\lambda_A) (\mu+\lambda_B)} \right) 
    \right] \\ 
    h(B) &=  \frac{\lambda_{tot}}{\mu + \lambda_{tot}} \left[ 
    \left( \frac{\lambda_B}{\mu+\lambda_A} \times \frac{\mu}{\lambda_{tot}} \right) 
    + p_b 
    \left( \frac{\lambda_A \lambda_B (2\mu+\lambda_{tot})}{\lambda_{tot}  (\mu+\lambda_A) (\mu+\lambda_B)} \right) 
    \right] .
\end{align*}
 We can also write those results in term of the market share $s_A$ and $s_B$ of both firms:
 \begin{equation} \label{equationsa}
     s_A= \frac{h(A)}{m}= \left( \frac{\lambda_A}{\mu+\lambda_B} \times \frac{\mu}{\lambda_{tot}} \right) 
    + p_a 
    \left( \frac{\lambda_A \lambda_B (2\mu+\lambda_{tot})}{\lambda_{tot}  (\mu+\lambda_A) (\mu+\lambda_B)} \right).
 \end{equation}

The market share of each firm will be the main quantity investigated in this study, we will sometimes call it the ``size'' of the firm.  \\

So far we have been agnostic on the meeting rate $\lambda_A$ and $\lambda_B$. If these meeting rates do not depend on the market shares $s_A$ and $s_B$, then Equation \eqref{equationsa} gives: 

$$ \frac{\partial s_A}{\partial p_a} = \frac{\lambda_A \lambda_B (2\mu+\lambda_{tot})}{\lambda_{tot}  (\mu+\lambda_A) (\mu+\lambda_B)} \leq 1, $$

the equality being only attained when $\mu=0$. Therefore a change in the agents preference $p_a$ will only partially affect the market share $s_A$, and the frictions tone down changes in preferences. However, the meeting rates may increase in each firm's market share and be functions $\lambda_A(s_A)$  of $s_A$.  \cite{bai2020search} shows, using an experiment on an online platform, that the the meeting rate does increase with the firm's past sales, and we provide evidence and a ground for this phenomenon in Section \ref{sec:heterogmeetingrates}. We encapsulate this idea by assuming from now on that the meeting rate might be an increasing function of the firm's market share. In this case, $ \frac{\partial s_A}{\partial p_a} $ becomes a much more complicated function, because an increase in $p_a$ affects $s_A$ both directly (conditionally on meeting both suppliers, more people prefer to match with $A$) and indirectly through affecting the meeting rates. This indirect effect can make $ \frac{\partial s_A}{\partial p_a} $ greater than $1$ (depending on the function $\lambda_A(s_A)$) so that frictions accentuate changes in $p_a$. \\



\vspace{0.3cm}

In the sub-sections thereafter we make different assumptions about the meeting rates $\lambda_A$ and $\lambda_B$. We first study, in sub-section~\ref{subsec:constantmeetingrate}, the case where the meeting rates are the same for both firms and do not depend on the firms market shares.  
We then study the case in sub-section~\ref{subsec:propmarketshare} where the meeting rate is proportional to the firm's market share. 
Finally,  we study in sub-section~\ref{subsec:propmarketshare} the case in which the meeting rate is an affine function of the firm's market share. 
In all three cases search frictions distort the allocation of agents to firms, as compared to the agents preferences. Yet, we show that in some cases, an increase in search frictions advantages the least preferred firm of the market, and in other cases the most preferred firm.


\subsection{Constant meeting rate}
\label{subsec:constantmeetingrate}

We assume in this sub-section that agents meet firms at a constant rate, i.e. a rate that does not depend on the firm. As a result, the agent meets $A$ and $B$ at rate $\lambda_A=\lambda_B=\lambda$, so that $\lambda_{tot} = 2 \lambda$.
Let us define the ratio $r_f = \mu / \lambda_{tot}$, which measures the intensity of search frictions. When this ratio is large, agents meet firms at a very slow rate, whereas when it is low, agents meet quickly both firms. Plugging those in Equation \eqref{equationsa}, we find that the market share of firm $A$, among matched agents is:
\begin{align*}
s_A= \frac{r_f}{1+2r_f} + \frac{p_a}{1+2r_f}.
\end{align*}


The first term grows with the intensity of search frictions and is uniform across firms, while the second one decreases with the intensity of frictions and is firm-specific. The market share of firm $A$ is not equal to the share of agents who prefer $A$ in general: $s_A\neq p_a$, except in the non-frictional case $r_f=0$. We also note that, when $ p_a > 1/2$, $s_A$ is always smaller than $p_a$: \emph{when meeting rates are constant across firms, search frictions always disadvantage the most preferred firm}. In fact, this effect gets stronger along with search frictions: the higher the search frictions, the lower the market share of the most preferred firm in the market. When the intensity of friction goes to infinity, both market shares converge to 50\%: the higher the frictions, the lower the impact of agents' preferences on the equilibrium allocation.  \\ 

We illustrate the results in Figure~\ref{imgconstantrate}, where we display how the market share $s_A$ of firm $A$  depends on the intensity of market frictions $\mu/\lambda_{tot}$. Note that in existing International Trade papers, 
agents all prefer the same firm, which provides the lowest quality-adjusted price on the market. \\


 \begin{figure}[h!]
    \centering 
    \includegraphics[width=0.7\linewidth]{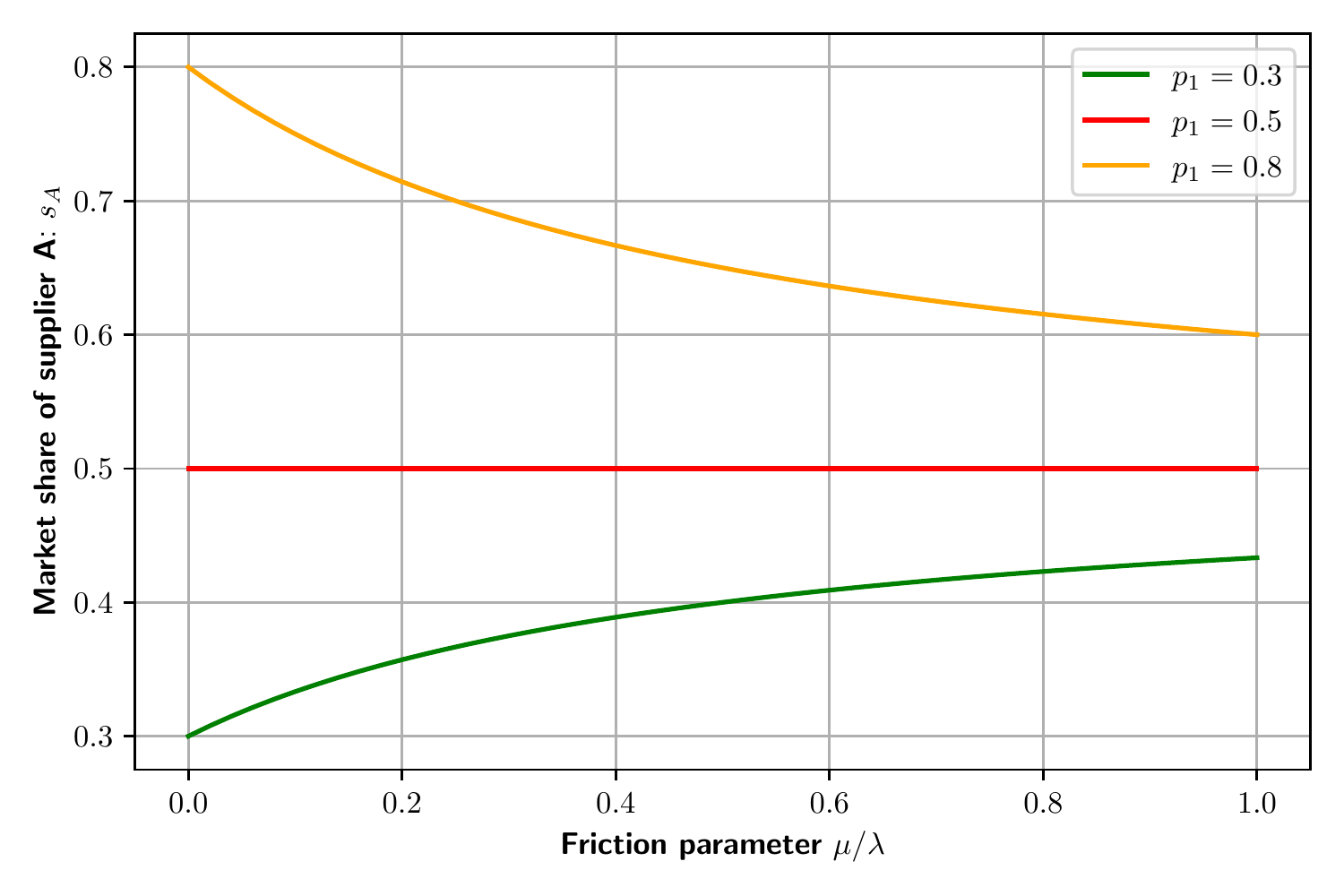}
    \caption{Constant meeting rate}
    \label{imgconstantrate}
\end{figure}


\subsection{Meeting rate increasing with market share}
\label{subsec:propmarketshare}

An alternative way to model search frictions is to model a meeting rate which increases in the market share of the firm. A rationale for this assumption is that firms with many agents are more visible to new agents. \\

We begin by a simple example where the meeting rate is proportional to the market share: $\lambda_A = \lambda s_A$ and $\lambda_B = \lambda s_B$.  This corresponds to the ``balanced matching'' assumption made in Labor economics. In order to justify such an assumption, let us consider for example the following mechanism: agents, in order to find a new firm, randomly pick an other agent in the market, search for its firm and meet with it. With this simple mechanism, the meeting rate of a given firm is directly proportional to its market share. \\

Plugging those meeting rates in Equation \eqref{equationsa}, we find that the market share of firm $A$, among matched agents is:

\begin{equation*}
    s_A =  \left\{
\begin{array}{ll} \displaystyle
      \displaystyle 0 & \text{when }  p_a \leq \frac{r_f}{ 2r_f+1} \\
      p_a (2 r_f+1) - r_f & \text{when }  \frac{r_f}{ 2r_f+1} \leq p_a \leq \frac{r_f+1}{ 2r_f+1} \\
      1 & \text{when } \frac{r_f+1}{ 2r_f+1} \leq p_a. \\
\end{array} 
\right. 
\end{equation*}

In this case, $s_A$ is piece-wise linear with respect to $p_a$. Figure~\ref{imgproprate} displays how the market share $s_A$ of firm $A$  depends on the intensity of market frictions $\mu/\lambda_{tot}$. As in the previous section, we note that $r_f=0$ leads to the frictionless equilibrium $s_A= p_a$.  However, we note important differences when compared with the previous section. \\

\textbf{Result 1: when meeting rates are proportional to the market shares, search frictions always advantage the most preferred firm.}
Indeed, when $ p_a > 1/2$, $s_A$ is always greater than $p_a$. We see that better by writing $p_a (2 r_f+1) - r_f= r_f(2p_a-1)+p_a$. This effect gets stronger along with search frictions: the higher the search frictions, the higher the market share of the most preferred firm in the market. In the previous case, frictions had an homogenizing effect, while in this case frictions accentuate existing inequalities. \\

\textbf{Result 2: when meeting rates are proportional to the market shares, a large enough friction intensity leads to a winner-takes-all scenario.}

Indeed we notice that when the intensity of friction reaches a certain level $r_f^*=p_b/(p_a-p_b)$, the firm $B$ disappears and the firm $A$ gets all the agents. \\

 Intense enough frictions can even cause one firm to disappear.  While frictions have been described in the literature as the force keeping small firms alive or advantaging them as in the previous section, we show that they can have the opposite effect and lead to the disappearance of small firms which would otherwise survive. We therefore claim that the way meeting rates depend on the size of firms matters and can change our representation of how frictions distort firm sizes and of who benefits from them. \\

\begin{figure}[h!]
    \centering 
    \includegraphics[width=0.7\linewidth]{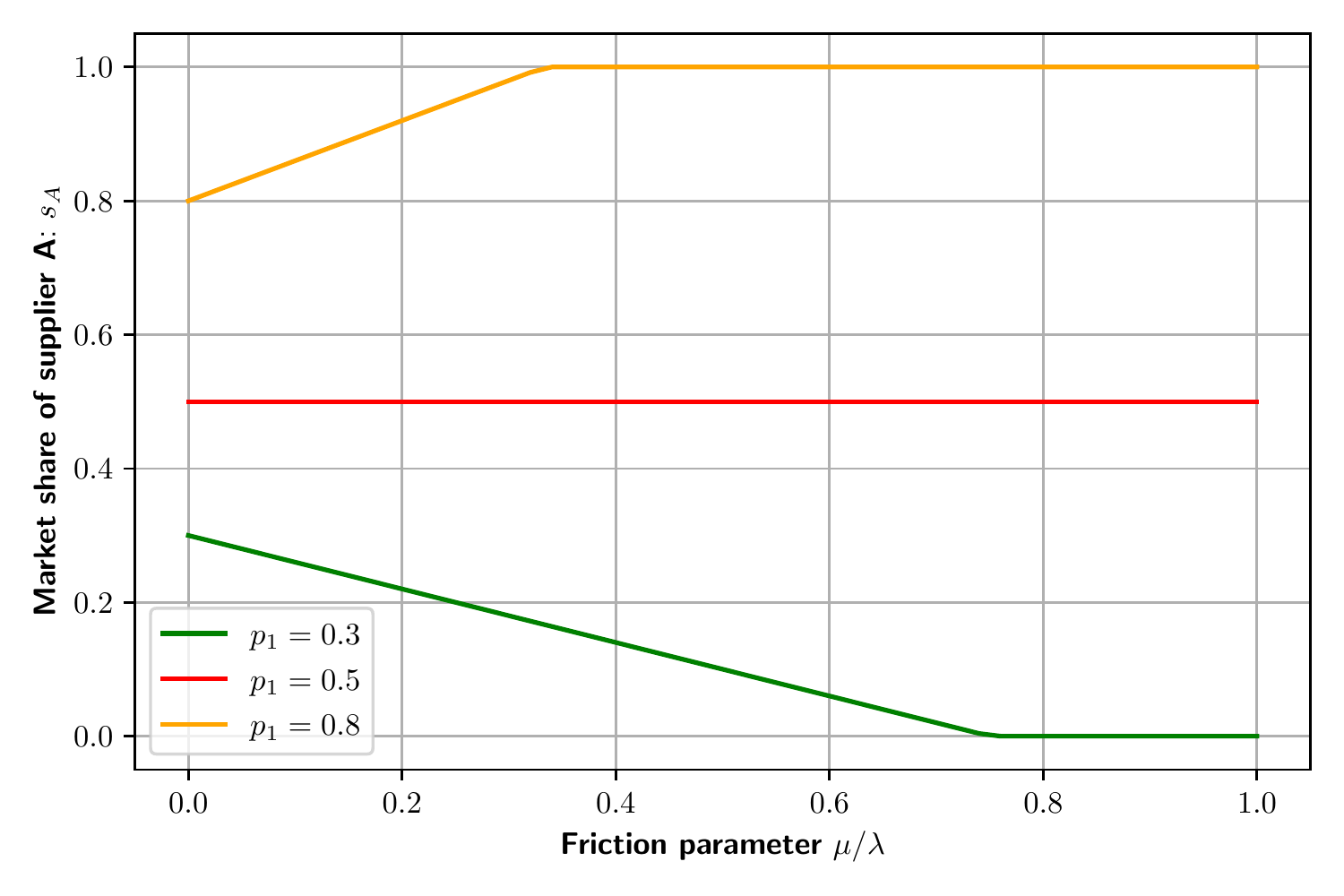}
    \caption{Meeting rate proportional to market share}
    \label{imgproprate}
\end{figure}


\subsection{Meeting rate increases in market share non-proportionally}
\label{sec:nonpropmarket_toy}

Now we allow for a more general model in which the meeting rate is not directly proportional to the market share of the firm. We assume that the meeting rate of each firm writes: $\lambda_i = \left( (1-\alpha)/2 + \alpha s_i \right) \lambda_{tot} $ for $i \in \{A,B\}$. $\alpha = 0 $ corresponds to the case, described in sub-section~\ref{subsec:constantmeetingrate}, where the meeting rate is constant and equal across firms. $\alpha = 1 $ corresponds to the case, described in sub-section~\ref{subsec:propmarketshare}, where the meeting rate is proportional to the firm's market share. \\

We have shown that in those two extreme cases, the effects of search frictions are opposed: for $\alpha=0$, higher search frictions advantage the least preferred firm, while for $\alpha=1$, higher search frictions advantage the most preferred firm. We want to show what happens in the case of $\alpha \in ]0,1[$. In this case, Equation \eqref{equationsa} has no simple closed form since it is a third-degree equation in $s_A$. However, using an analytical solver allows to show how frictions affect the market shares for different values of $\alpha$. \\

Figure~\ref{fig:saVSp1} illustrates how the market share of a firm varies with the proportion $p_a$ of agents preferring this firm for different values of $\alpha$, for a given value of frictions intensity (here $\mu/\lambda =0.7$). We first note that the market share of the most preferred firm is increasing with $\alpha$: the most preferred firm benefits from a higher $\alpha$. The second thing we note is that for some values of $\alpha$, the effect of frictions is ambiguous: it does not \emph{always} advantage the most preferred agent nor always advantage the least favorite one. For instance, for $\alpha= 0.85$, and for $\mu/\lambda =0.7$, frictions advantage firm $A$ for $0.5 \leq p_a \leq 0.9$: in this range the yellow curve is above the grey line, so $s_A >p_a$. However, we note that for  $p_a \geq 0.9$, frictions advantage the least preferred firm - firm $B$ in this example. The effect of friction is thus ambiguous in that case. Theorem \ref{alwaysloser} characterizes the cases where the effect of frictions is ambiguous. 

\begin{theorem} \label{alwaysloser}
Take $\alpha \in [0,1[$. \\ 
When $r_f \geq \frac{\alpha-1/2}{1-\alpha} $, then $s_A \leq p_a$ for all $p_a \geq 0.5 $: frictions have a homogenizing effect.\\
When $r_f < \frac{\alpha-1/2}{1-\alpha} $, then $s_A > p_a$ for at least some $p_a \geq  0.5 $: frictions may accentuate existing inequalities.
\end{theorem} 

We note that, for $\alpha=0$, we always have $r_f \geq \frac{\alpha-1/2}{1-\alpha}$, thus we recover the result from Section~\ref{subsec:constantmeetingrate} that frictions always have a homogenizing effect in that case, while for $\alpha \rightarrow 1$ one always have $r_f < \frac{\alpha-1/2}{1-\alpha} $, which is coherent with the results from section~\ref{subsec:propmarketshare}. One important insight of this Theorem is the following: for any $\alpha \leq 1/2$, \emph{frictions always have a homogenizing effect}. Frictions only amplify inequality when $\alpha > 1/2$, and we argue in Section \ref{sec:heterogmeetingrates} that this is often the case in practice. \\

Figure~\ref{fig:saVSfriction} illustrates how the market share of a firm varies with the intensity of frictions, for different $\alpha$. We take the case of a firm that is preferred by 10\% of agents. Once again we notice that the market share of the least preferred firm decreases with $\alpha$. For small $\alpha$ the market share of the firm increases with the frictions' intensity (like in the case described in sub-section~\ref{subsec:constantmeetingrate}), but for  large $\alpha$ the market share decreases with $\mu/\lambda$ (like in sub-section~\ref{subsec:propmarketshare}). We again see that for some values of $\alpha$ (for instance for the yellow curve representing $\alpha = 0.85$), the market share is non-monotonic in the intensity of frictions, thus the effect of frictions is ambiguous. The very direction of the effect (whether it is an advantage or disadvantage for the firm) depends on the intensity of the friction: the curve is below the line $ s=0.1$ for some values of $r_f$ and above for others. 

\newpage

\begin{figure}
    \centering
    \includegraphics[width=0.62\linewidth]{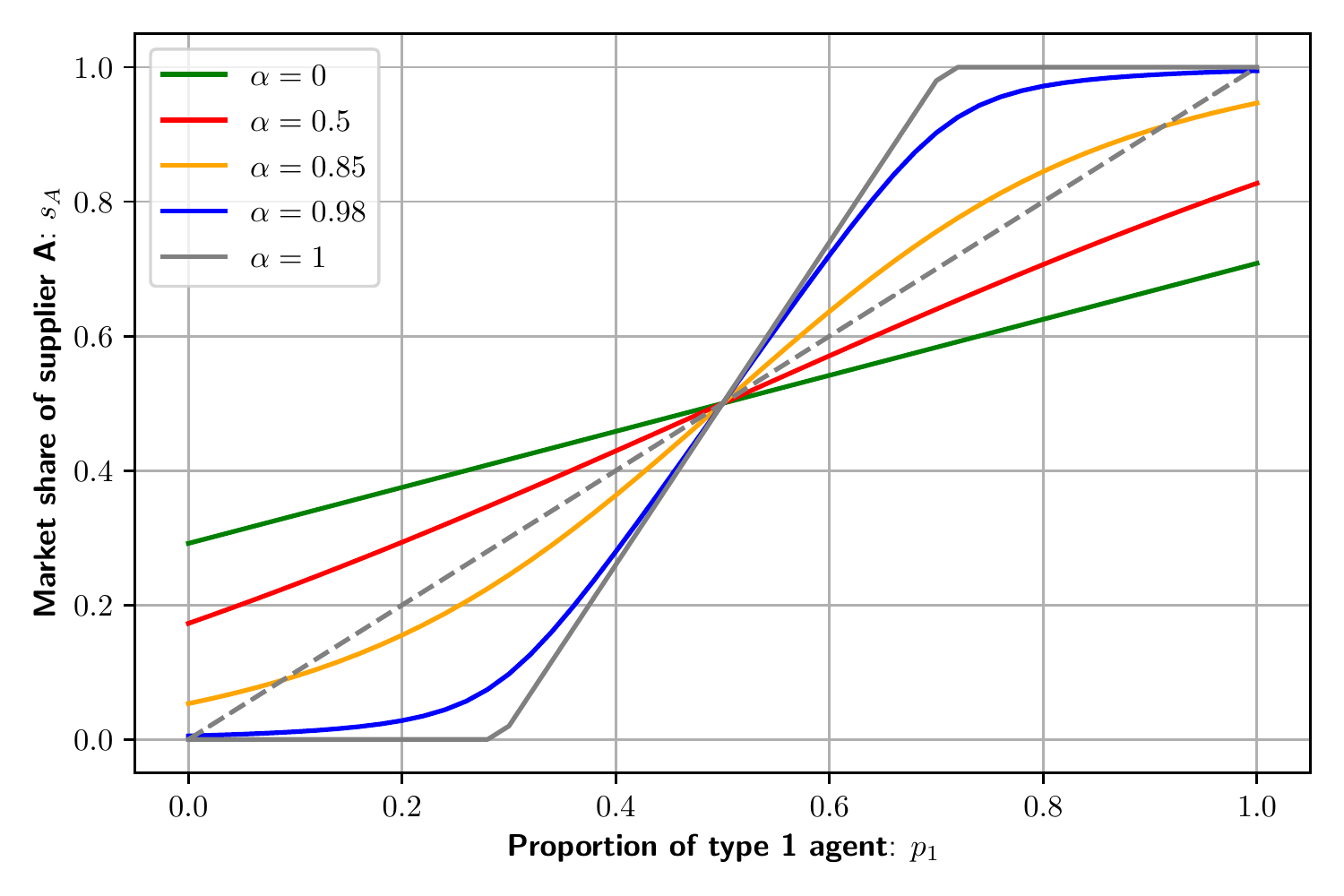}
    \caption{Market share depending on proportion of type $a$ agent, for different $\alpha$. Graph built with $\mu/\lambda=0.7$.}
    \label{fig:saVSp1}
\end{figure}
\begin{figure}
    \centering
    \includegraphics[width=0.62\linewidth]{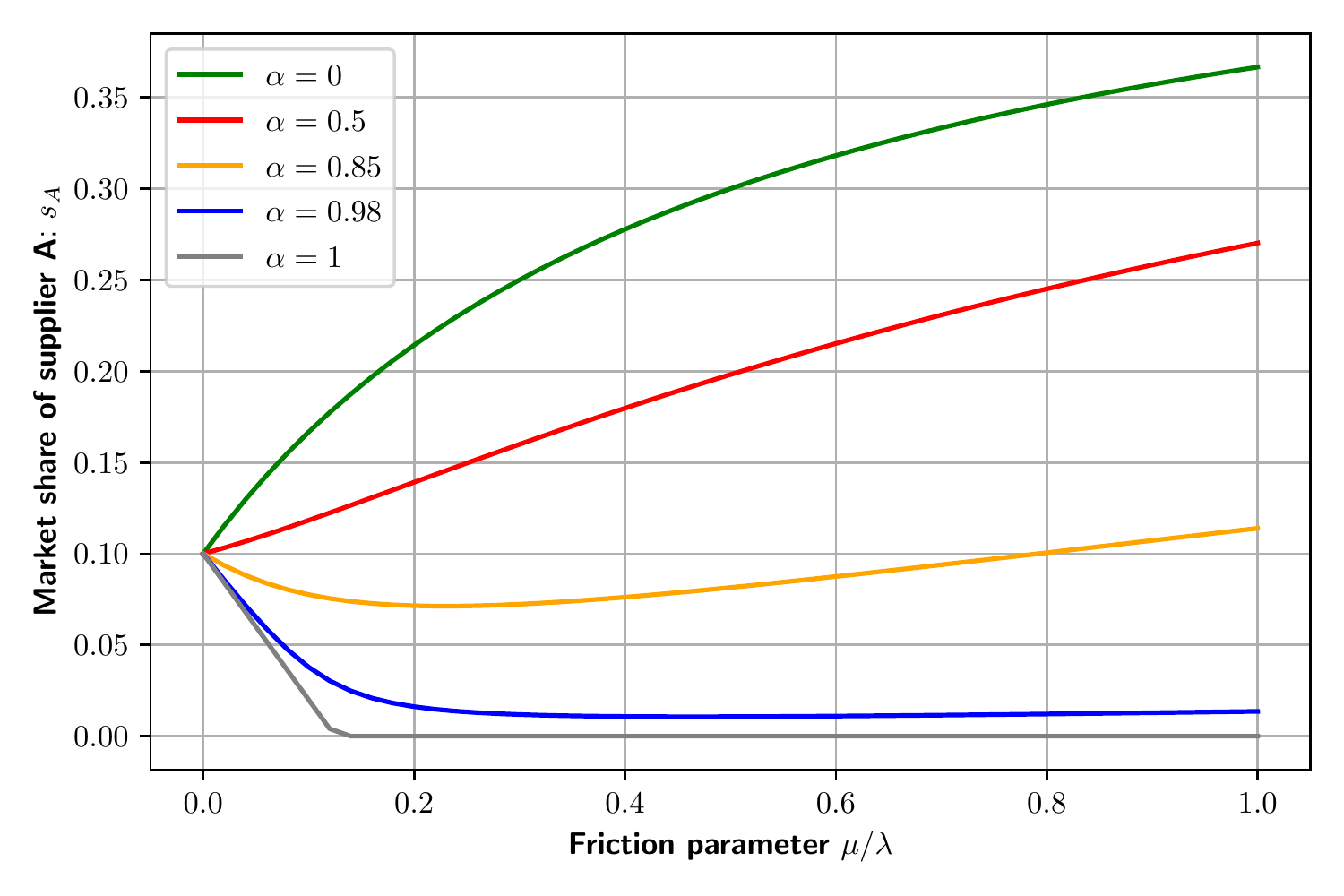}
    \caption{Market share depending on friction parameter, for different $\alpha$}
    \label{fig:saVSfriction}
\end{figure}

 \newpage


\section{Multi-firms model}
\label{sec:multifirms}
We now show that our results on the effects of frictions generalize to an environment with a continuum of firms. We characterize how the firm size distribution is affected by search frictions in presence of heterogeneous agents' preferences, and according to the shape of frictions, when there are more than two firms. We describe which intuitions of our toy model remain valid in that set-up.  We use a terminology close to \cite{lindenlaub2016multidimensional}, studying what happens when one allows meeting rates to vary across firms. \\ 

\subsection{Setting} \label{sec:multisetting}

There is a continuum of firms where each firm is indexed by its type $y$ and which is uniformly distributed on the interval $[0,1]$. We assume that firms have no capacity constraint and that all matches yield a positive surplus. Therefore firms are always happy to match with any agent that they encounter. \\ 

We also have a continuum of agents, which can be either matched to a firm or unmatched. Each agent is indexed by its type $x$. We denote $\ell(x)$ the total measure of type $x$ agents, and we impose the normalisation $\int \ell(x)dx=1$. We denote $u(x)$ the measure of unmatched  type $x$  agents.  Whether matched or unmatched, agents exit the market at Poisson rate $\mu$. If the agent is unmatched, she meets firms of type $y$ at rate $\lambda_0(y)$ and match with them. We allow for on-the-job search: if the agent is already matched with a firm $y$, she meets an offer from firm $y'$ at rate $\lambda_1(y')$, and accepts this offer only if it is better than her current offer. In what follows, we will make the hypothesis that, for all $y$, $ \lambda_0(y)= K \lambda_1(y)$ for some constant $K$.\footnote{This hypothesis allows us to more clearly show the impact of the way we model meeting rates on the distribution of firm's sizes. Moreover, having similar distributions of meeting opportunities across firms regardless of whether you are matched or unmatched is realistic. \cite{lindenlaub2016multidimensional} make a similar assumption when stating that employed and unemployed workers sample job offers from the same sampling distribution, but for them this distribution is fixed and exogenous while we will allow meeting rate to depend on the size of firm $y$. \\ }  We therefore can drop the subscript and write $\lambda_1(y)=\lambda(y)$.  \\ 

We denote $\sigma(x,y)$ the value that agents get from a type-$(x,y)$ match. A agent $x$ matched with a firm $y$ accepts an offer from $y'$ if and only if $ \sigma(x,y) < \sigma(x,y')$.  We are agnostic about the form of this function $\sigma(\cdot, \cdot)$, which is exogenous and given in our model. The fact that $\sigma$ varies across $x$ means that different types of agents have different preferences: constant $\sigma(\cdot,y)$ for all $y$ would imply homogeneous preferences.

In much of what will follow we will take the following example: we will take $x$ and $y$  in $ [0,1]$ and we will consider $\sigma(x,y)= f(d(x,y))$ where $d(x,y)=\min_{k\in \mathbb{Z}}|x-y +k|$ and $f$ is any strictly decreasing function so that agents of type $x$ have a preference for firms that have a type close to $x$. \footnote{The modulo in the distance is added to take into account the corner problem described in \cite{gautier2006labor} and to avoid an asymmetry between firms placed in the corners (near $0$ and $1$) and the other ones. $x$ and $y$ can therefore be thought of as a location on a circle of circumference $1$, so that $x=1$ is equivalent to $x=0$ and the same for $y$ (see \cite{gautier2006labor}).}\\

 Our analysis is in partial equilibrium: 
 the distribution of types is fixed once and for all and is exogenous; $\ell(x)$ is given and firms type are uniformly distributed on the interval $[0,1]$. In this model, there is neither an endogenous entry of firms nor an adaptation of firms to agents' preferences. \\

Just as in the previous section, we want to compute the expected number of agents that are matched with a given firm $y$ that we denote $h(y)$. In order to compute $h(y)$ we first study the steady-state equilibrium density of type $(x,y)$ matches, denoted by $h(x,y)$, which indicates how many agents of type $x$ match with firms of type $y$. \\

\subsection{Steady-State distribution of matches and firm sizes}

The inflow into the stock of type $(x, y)$ matches is composed of two groups: unmatched type $x$ agents who draw a type $y$ firm with Poisson rate $\lambda_0 (y)$ and accept it; and type $x$ agents already matched in any type $y'$ job who draw a type $y$ offer at rate $\lambda_1 (y)$ and accept it.

The outflows from the stock of type $(x,y)$ matches are: the agents which leave the market, which happens at rate $\mu$, and the agents which meet a `better' firm. Therefore in equilibrium:
\begin{align} \label{eq:master}
h(x,y) \times \underbrace{ \left( \mu + \int \mathbbm{1}[\sigma(x,y')>\sigma(x,y)] \lambda_1(y') dy' \right) }_{G(x,y)} = \nonumber \\
\lambda_0(y) u(x) + \lambda_1(y) \int    \mathbbm{1}[\sigma(x,y)>\sigma(x,y')] h(x,y') dy'.
\end{align}
We denote $G(x,y)$ the rate at which type $(x,y)$ matches are destroyed, either due to agent leaving the market or to another offer being accepted. \\

Writing similar inflow/outflow equations for the stock of type $x$ unmatched agents gives:
$$  u(x)(K\lambda_{tot}+\mu)= \ell(x) \mu,$$

where $K \lambda_{tot}=  K \int_0^1 \lambda_1 (y) dy= \int_0^1 \lambda_0 (y) dy$. 

After a few computations, we show that:
\begin{equation}\label{hxy}
h(x,y)= \frac{u(x) \lambda_1(y) K(\mu + \lambda_{tot})}{G(x,y)^2}   = \frac{\mu K(\mu + \lambda_{tot})}{K\lambda_{tot}+\mu}\frac{\ell(x) \lambda_1(y)}{G(x,y)^2} .
\end{equation}

Note that when taking a constant $\lambda_1$ we recover the result from \cite{lindenlaub2016multidimensional}. \\

The mass $h(y)$  of agents matched with firm $y$ writes: 
\begin{equation*}
h(y) = \int h(x,y) dx =  \frac{\mu K(\mu + \lambda_{tot})}{K\lambda_{tot}+\mu} \lambda_1(y) \int \frac{\ell(x) }{G(x,y)^2} dx   .
\end{equation*}

As in the previous section we will denote $m$ the mass of matched agents and  we want to find  the ``market share'' of the firm indexed by $y$. As we are in a continuum, each firm has a market share of measure $0$, but we can still measure relative market shares by considering the following differential: 

\begin{equation} \label{eq:formulas}
   s(y)= \lim_{dy\rightarrow 0} \frac{1}{m}\frac{\int_y^{y+dy}h(u)du}{dy}= \frac{\mu (\mu + \lambda_{tot})}{\lambda_{tot}}\lambda_1(y) \int \frac{\ell(x) }{G(x,y)^2} dx. 
\end{equation}

This will be the main quantity of interest for the rest of this section and that can be seen as an infinitesimal rescaled equivalent of the market share. We sometimes call this quantity ``size'' of firm $y$.\footnote{Informally, when the number of firms is finite and equal to $N$, $s(y) \simeq \tilde{s}(y)\times N$ where $\tilde{s}(y)$ is the true market share of $y$. As $\tilde{s}(y)$ goes to zero when $N$ goes to infinity, we have to consider its rescaled version $s(y)$. Note that $s(y)$ can be greater than 1.  }

Remark: The market share of firms does not depend on the factor $K$. \\ 


\textbf{Result 1: When the frictions are negligible, the size of a firm is proportional to the fraction of agents who rank this firm as their favorite.} 

When the ratio $ r_f = \mu/\lambda_{tot} $ which measures the intensity of frictions goes to zero, then $ s(y) \rightarrow \ell(y) $. We prove this result in Annex \ref{annex_sec_multifirms_part2}.
This result is consistent with the results of the toy model with two firms. Without any friction, the firm size distribution simply corresponds to the preferences' distribution.

\subsection{Constant meeting rates} \label{subsec:multiconst}
In this section we take $ \lambda_1(y)= \lambda_{tot}$ so that $\lambda_1$ does not vary across firms. Then Equation \eqref{eq:formulas} simplifies into:  
\begin{align*}
 s(y) &=  \int \frac{r_f (r_f + 1)\ell(x)}{(r_f+\int_{y'}\mathbbm{1}[\sigma(x,y')>\sigma(x,y)] dy')^2} dx \\
&=  \int \frac{r_f (r_f + 1)\ell(x)}{(r_f+2d(x,y))^2}dx = \ell* K (y). 
\end{align*}

 

 $s(y)$ is the convolution of $\ell$ and $K$ a kernel $K(t)=\frac{r_f (r_f + 1)}{(r_f+2d(t))^2}$ with $d(t)=\min_{k\in \mathbb{Z}}|t +k|= \min(t, 1-t)$. Therefore $s(y)$ behaves as a weighted moving average of $\ell$ around $y$, with a characteristic bandwidth proportional to $r_f$. Since it is an average, it behaves more smoothly than $\ell$.\\

  When $r_f$ goes to $0$, the bandwidth shrinks. The kernel $K$ becomes close to a Dirac function and $ s(y) \rightarrow l(y) $. This result does not come as a surprise and is in line with results from Section \ref{sec:twosup}. When there is no friction, each agent is matched to its favorite firm, and thus the size of a firm is proportional to the fraction of agents who rank this firm as their favorite.\\

 \textbf{Result 2: When the friction intensity goes to infinity, firms sizes are uniformly distributed.}
 
 On the contrary, when $r_f$ goes to infinity, $K(t)\rightarrow 1$ for all $t$, the kernel becomes constant, and each firm has the same size in the limit. 
 Once again this is no surprise: in this setup, when frictions are too strong, agents only have time to meet one firm, that they pick uniformly among all firms. This is consistent with the results of section \ref{subsec:constantmeetingrate}, with constant meeting rates and two firms. \\

 \textbf{Result 3: With constant meeting rates, frictions have an homogenizing effect.} 
 
 In presence of heterogeneous preferences, what we mean by having an homogenizing effect needs to be defined formally. We find that, in presence of search frictions: 
 
 \begin{itemize}
     \item $\max_y s(y) \leq \max_y l(y)$: the biggest firm has a market share that is lower than it would have without frictions. The same way, $\min_y s(y) \geq \min_y l(y)$: frictions have an averaging effect that smooth extreme values.
     \item $\Var_y(s(y)) \leq \Var_y(l(y))$: this smoothing leads to a reduction of the variance as compared to the frictionless case.
 \end{itemize}
 
 All of this is due to the fact that $s$ behaves as a moving average of $\ell$. We take an example to get a sense of the quantitative effect of frictions in that set up. Let us take for instance $l(x)=5 \times \mathbbm{1}[x \in[0.4,0.6]]$. Figure~\ref{fig:syVSyCONST} shows how the distribution $s(y)$ of firm size changes with the intensity of frictions $r_f$. The green line represents the distribution of agents' preferences. When frictions are non-existent -i.e. when $r_f = 0$ - the size distribution is simply the preferences distribution. The figure shows that the higher the frictions - i.e. the higher $r_f$ - the smoother the distribution. Figure~\ref{fig:varVSrfCONST} completes this picture by showing that the variance of the firm size distribution $s(y)$ decreases with the intensity of frictions. \\

This homogenizing effect of frictions is in line with the results of the toy model with constant meeting rates. Introducing heterogeneous preferences therefore does not affect the impact that search frictions have on the firm size distribution, when meeting rates are constant. We show in the next sub-section that when meeting rates increase in the firms' market shares, the results change dramatically.

\begin{figure}
    \centering
    \includegraphics[width=0.62\linewidth]{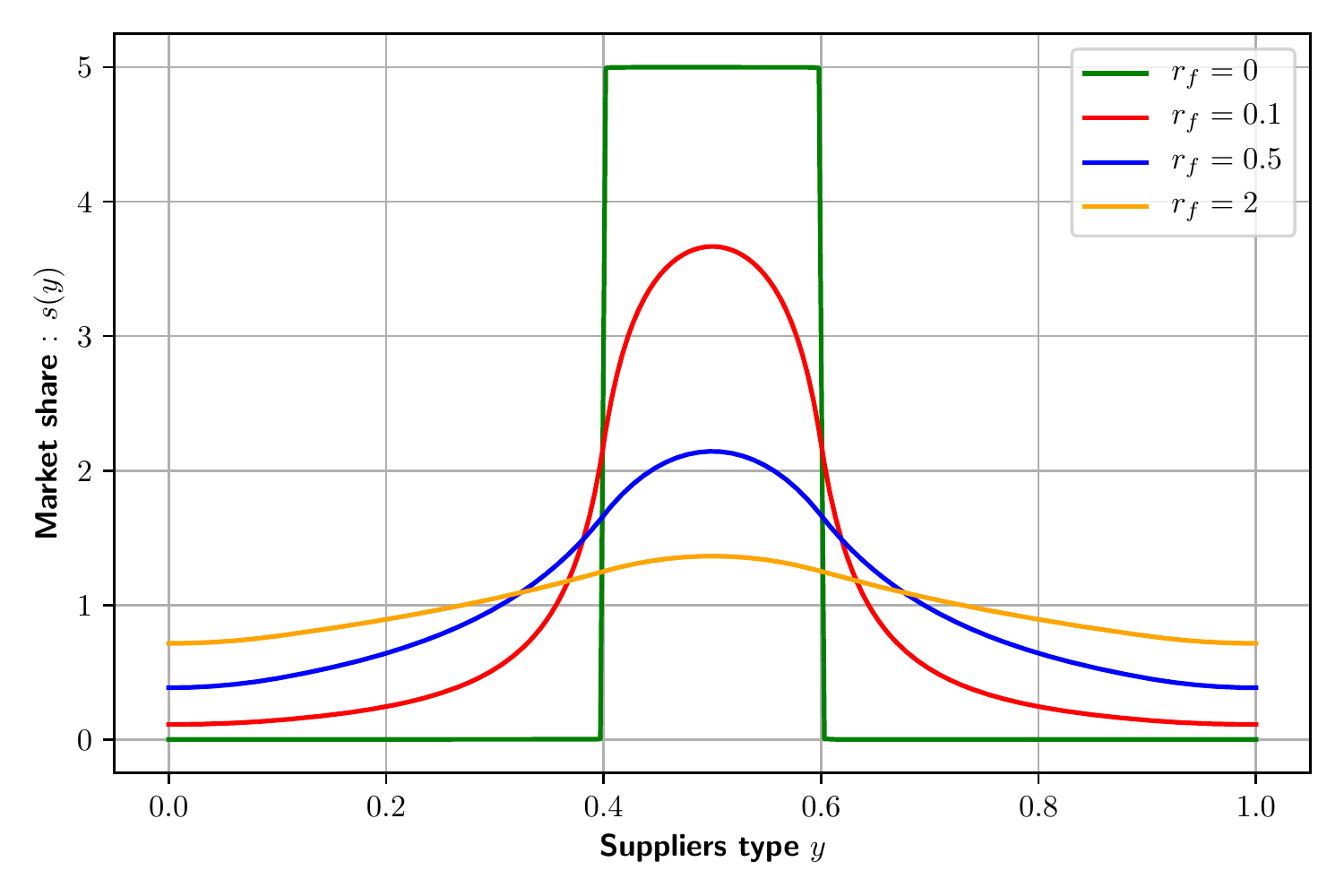}
    \caption{Market share depending on firms types, for different  $r_f$. }
    \label{fig:syVSyCONST}

    \centering
    \includegraphics[width=0.62\linewidth]{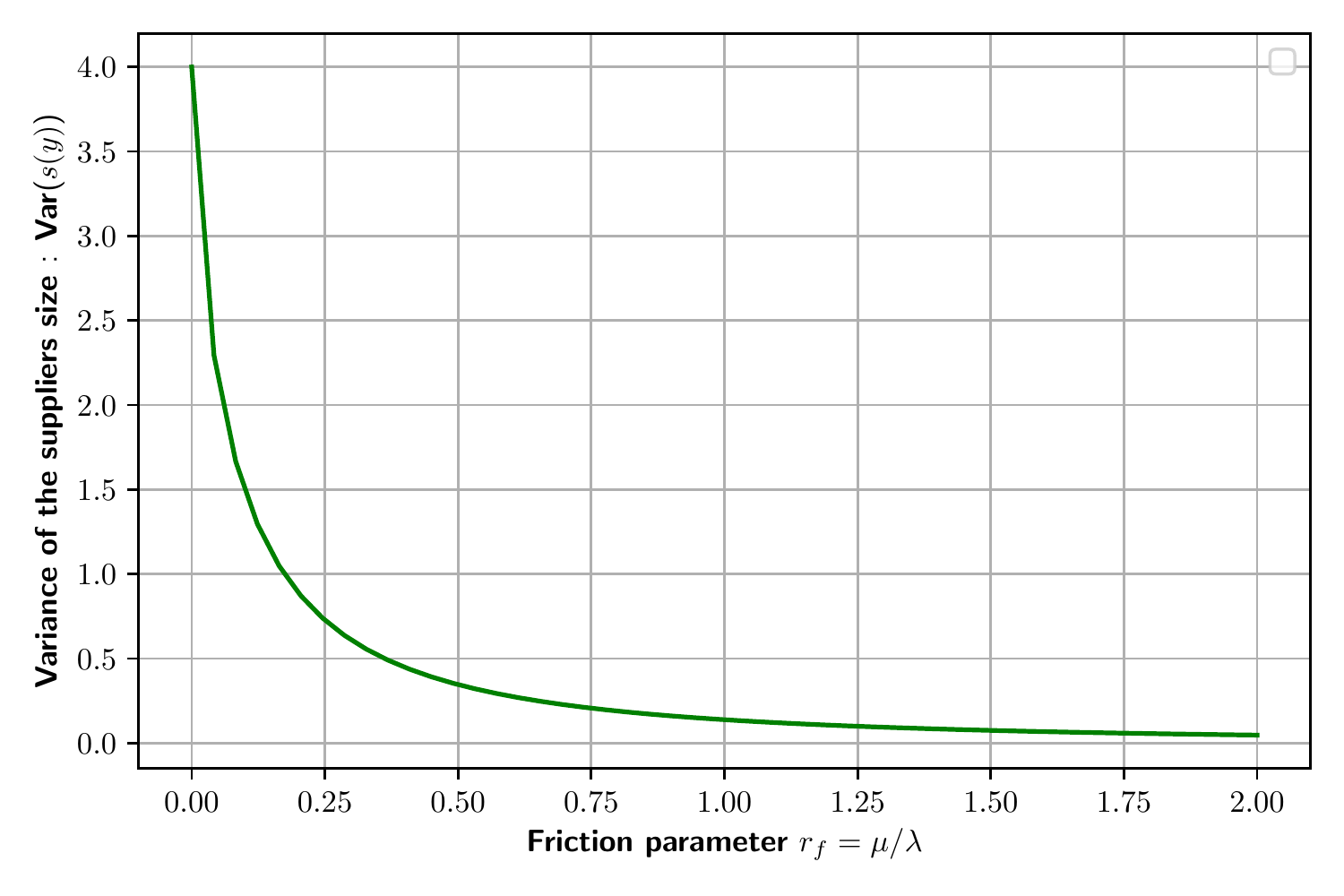}
    \caption{Variance of the market share depending on $r_f$.}
    \label{fig:varVSrfCONST}
\end{figure}

\subsection{Meeting rate proportional to the market share}

We now consider the case where the meeting rate of a firm is directly proportional to its size, as in Section \ref{subsec:propmarketshare}. We assume that for a given firm $y$: $\lambda(y)= \lambda_{tot}  \times s(y) $. 


The main equation on $s(y)$ is more complicated than in the previous case because it is implicit. Indeed Equation~\ref{eq:formulas} gives, for all $y$ so that $s(y) >0$,
$$\int \frac{r_f (r_f + 1)\ell(x)}{(r_f+\int_{y'}s(y')\mathbbm{1}[\sigma(x,y')>\sigma(x,y)] dy')^2} dx=1.$$

Let us first note that, like in the previous case, when the frictions are negligible, the size of a firm is proportional to the fraction of agents who rank this firm as their favorite. We will now present a first result that claims that, as in the two firms case, a winner-takes-all scenario will prevail when the friction intensity is high enough.

Let us denote $  \mathcal{S}$ the set of every interval of size $1/2$ on the interval $[0,1]$ regarded as a circle. Let us denote $\mathcal{S}_y=\{I \in | y \notin I\}$ the set of intervals in $\mathcal{S} $ that do not contain $y$.

\begin{theorem} \label{theo:concentration}
 If there is $y^*$ so that $\int_I \ell(x) < 1/2$ for every interval $I$ in $ \mathcal{S}_{y^*}$, then, when $r_f$ goes to infinity, the distribution of $s(y)$ converges to the Dirac function $\delta_{y^*}$. 
\end{theorem}

\textbf{What does Theorem \ref{theo:concentration} say ?}
Theorem \ref{theo:concentration} means that for all interval $I$ that does not contain  $y^*$, $\int_I s(y)dy=0$ for large enough frictions. We have just as in Section \ref{subsec:propmarketshare} a concentration of the market in one particular point while all other points are being dried out. Theorem \ref{theo:concentration} characterizes the point of concentration $y^*$. It is not in general the point $x_m$ where $\ell$ reaches a maximum. We note that $x_m$ is the point where the size reaches a maximum in frictionless cases, meaning that in this setup the \emph{friction intensity changes the ranking of firm sizes.} The point $y^*$  is the ``median point'' of the distribution $\ell$: the intervals $[y^*, y^*+1/2]$ and $[y^*, y^*-1/2]$ (regarded as intervals over a circle) both contain half of the mass of the distribution $\ell$.\footnote{ In which cases doesn't Theorem~\ref{theo:concentration} apply ?
In some degenerate cases, there might not be any points so that $\int_I \ell(x) < 1/2$ for every interval $I$ in $ \mathcal{S}_{y^*}$. This the case for instance when $\ell$ is constant or when $\ell(x)=1+\cos(2x/(2 \pi))$. In the first case, the perfect symmetry of the problem ensures that firms sizes remain equal regardless of the frictions. The second case is not covered by the theorem, and there again the perfectly similar conditions on points $x$ and $x+1/2$ would certainly prevent the concentration of market share in only point. If there are strong oscillations in $\ell$, Theorem \ref{theo:concentration} may not be applicable.} Simple examples where Theorem~\ref{theo:concentration} applies include cases where $\ell$ is unimodal. \\ 

The most important take-away of Theorem \ref{theo:concentration} is that, even when considering a continuum of firms, the way meeting rates depend on firms size matters. As in the case with two firms, a concentration phenomenon as well as a winner-takes-all phenomenon happen in this setting when frictions go to infinity. This is to our knowledge the first time that such a phenomenon is derived as the consequence of search frictions.


This is not possible to obtain a closed-form expression for $s(y)$ in this case. In the next sub-section, we go beyond Theorem 2 with the help of numerical approximations and provide further evidence on the effect of frictions on the firm size distribution.


\subsection{Numerical evidence} \label{subsec:numerical}


We now take $\lambda(y)=\lambda_{tot}(\alpha s(y)+(1-\alpha) )$.\footnote{We normalize here the mass of firms to one.} We compute numerically the value of $s$ for different values of $\alpha$ (see Annex \ref{annex:numerical}). \\

First, we study a variation of the previous example with $\ell(x)=2 \times \mathbbm{1}[x \in[0.25,0.75]]$. This example is one of the simplest one can think of as agents' preferences are single-peaked. We showed in Figure~\ref{fig:syVSyCONST} that with constant meeting rates, higher frictions lead to a smoother firm size distribution. We exhibit in Figure~\ref{fig:syVSyPROP1a} the equilibrium firm size distribution, when we take $\alpha = 0.8$. \textbf{The effect of search frictions is non-motonic:} search frictions first concentrate the preferences' distribution - blue and red curves - then attenuate it - yellow curve. This result is in line with results from Section \ref{sec:twosup}: when $0<\alpha<1$, an increase in frictions does not always have the same effect: it can concentrate the distribution (for small values of friction intensity) or smooth it (for higher values). The threshold between those two behaviours depend both on the value of $\alpha$ and on the shape of the preference distribution $\ell$. \\

In Figure~\ref{fig:syVSyPROP1b}, we display a similar similar Figures with a higher dependency of the meeting rate on the firm's market share: $\alpha = 0.99$. We see that with this higher $\alpha$, the firm size distribution is more concentrated for every level of friction intensity $r_f$ when compared with $\alpha = 0.8$. For instance for $r_f=3$, the effect is smoothing for  $\alpha = 0.8$ while it is concentrating for $\alpha = 0.99$. From a number of numerical analysis conducted, we formulate the following hypothesis: \textbf{higher values of $ \alpha$ lead to more concentrated distributions of $s$}, everything else remaining equal. While this seems hard to prove formally as  $s$ has no closed-form expression, it is an outcome that arises repeatedly in our numerical studies.   \\

Second, we take the example of double-peaked preferences to illustrate the fact derived in the previous Section that the firms benefiting from frictions can be different from the most preferred firms where $\ell$ reaches a maximum, even in cases where frictions concentrate the distribution (meaning they decrease its variance). Theorem \ref{theo:concentration} states that the maximum of $s$ is reached at the median firm, where all agents tend to concentrate when $\alpha=1$. Figure~\ref{fig:syVSydoublepic}  shows such a phenomenon arises even for smaller values of $\alpha$. Here the median firm has a size bigger than $\max_y \ell(y)$, which can never occur when $\alpha=0$ (as mentioned in Section \ref{subsec:multiconst}).  \\

\begin{figure}
\begin{subfigure}[t]{0.45\textwidth}
        \centering\includegraphics[width=\linewidth]{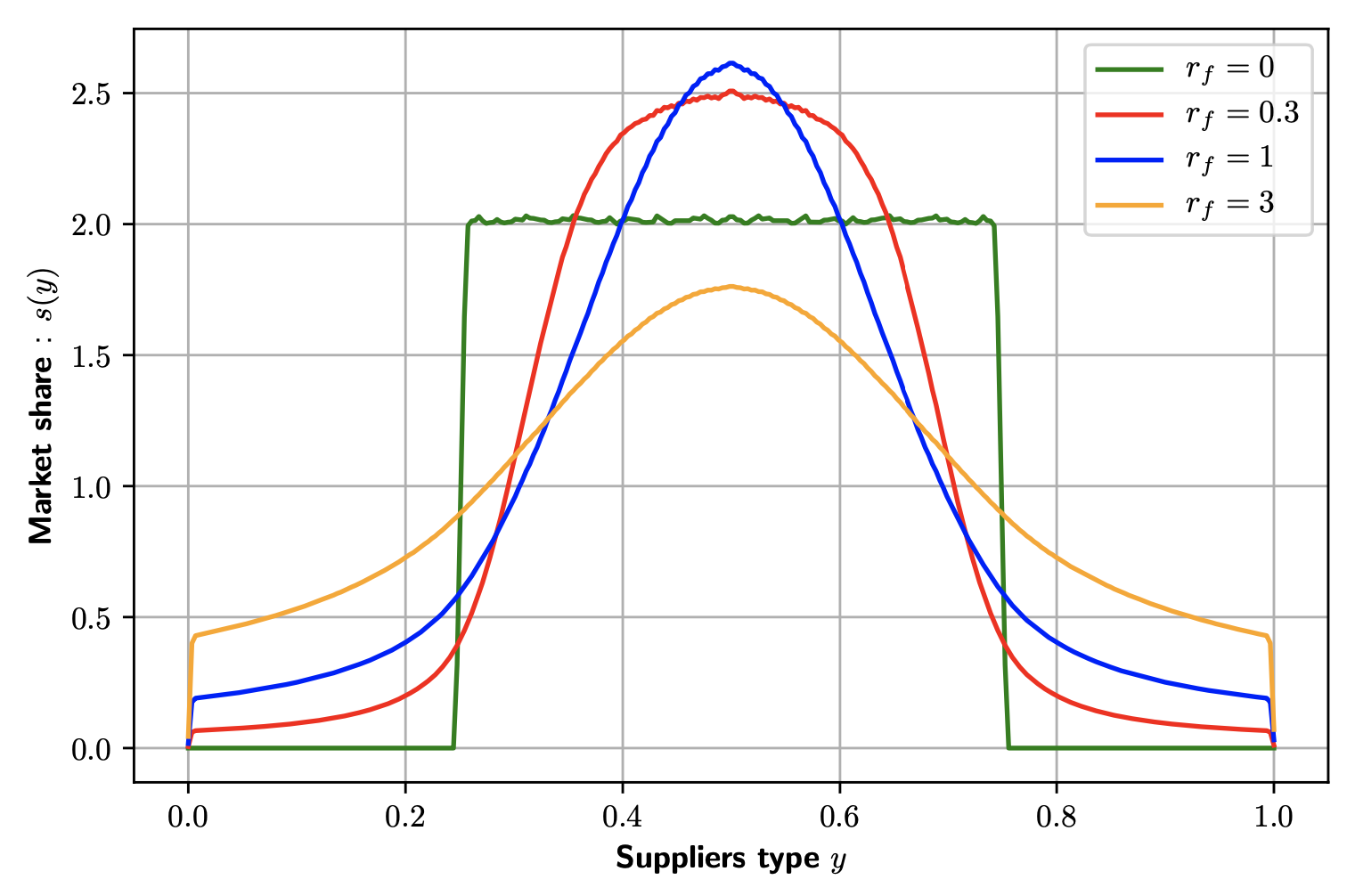}
        \caption{Market share depending on firms type, for different $r_f$, with $\alpha = 0.8$.}
        \label{fig:syVSyPROP1a}
\end{subfigure}
\hfill
\begin{subfigure}[t]{0.45\textwidth}
        \centering\includegraphics[width=\linewidth]{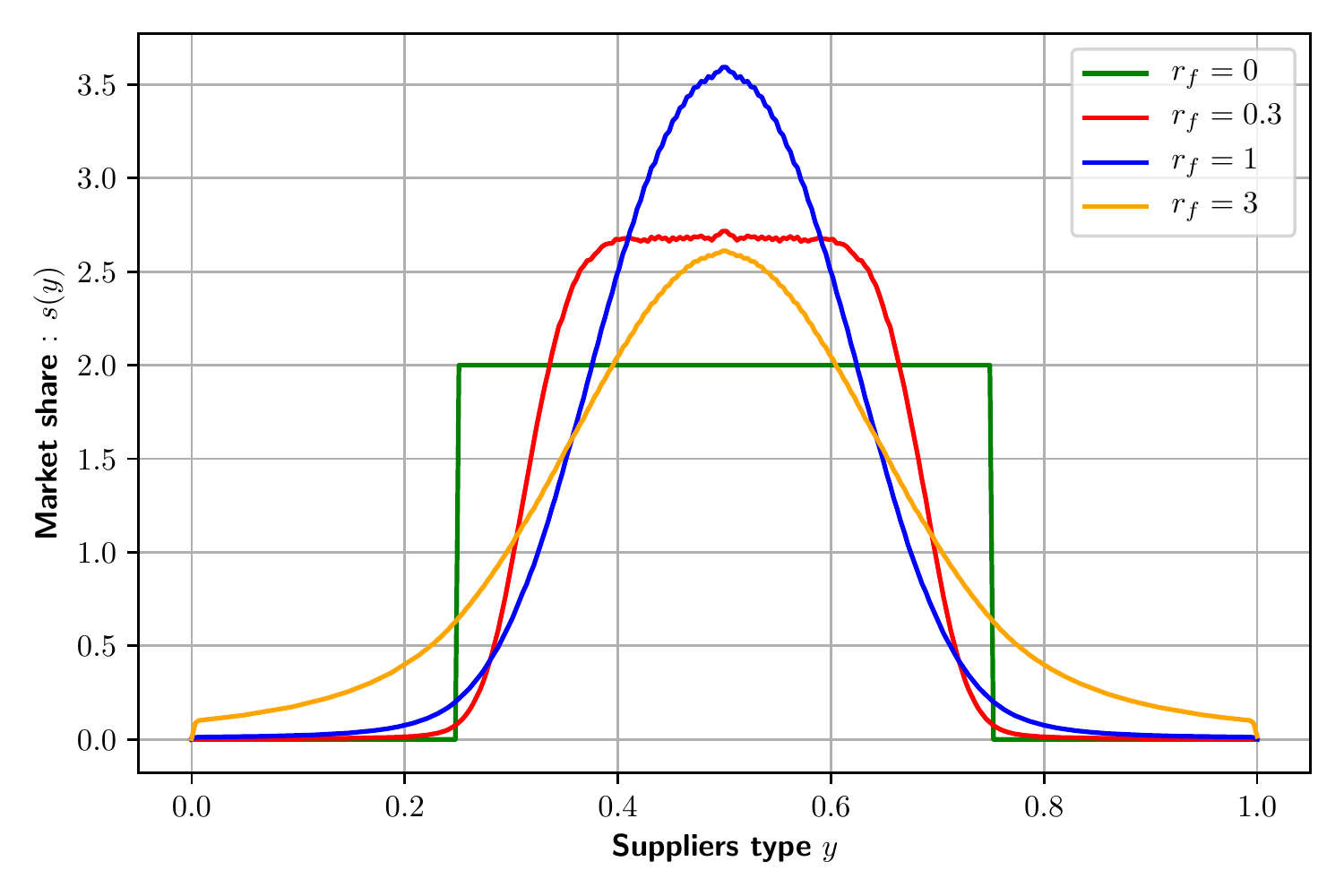}
        \caption{Market share depending on firms type, for different $r_f$, with $\alpha = 0.99$.}
        \label{fig:syVSyPROP1b}
\end{subfigure}        
\end{figure}

\begin{figure}
        \centering\includegraphics[width=0.6\linewidth]{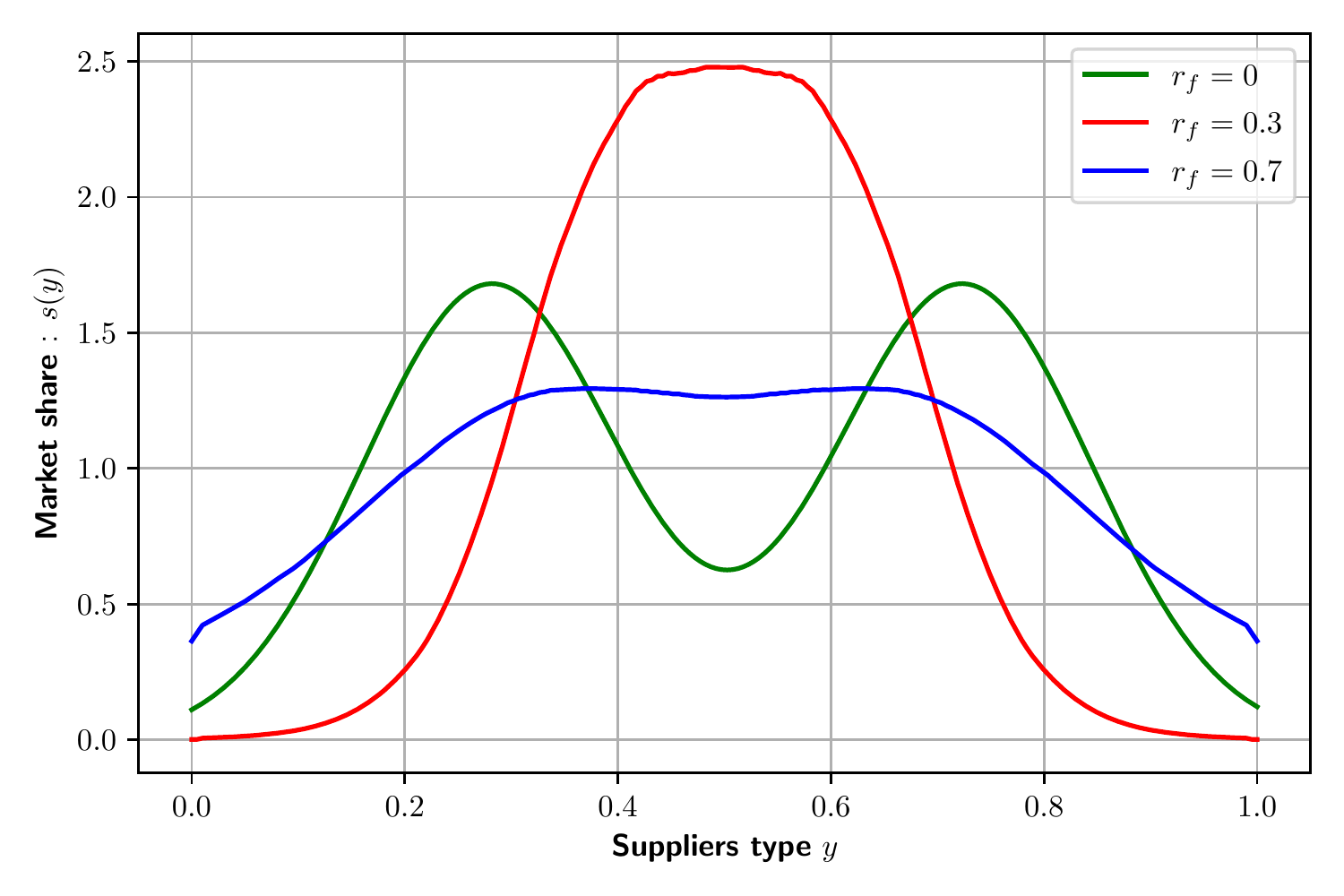}
        \caption{Market share depending on firms type, for different $r_f$, with $\alpha = 0.95$.}
        \label{fig:syVSydoublepic}
\end{figure}

\subsection{Optimal $\alpha$ and externalities}

After observing that the value of $\alpha$ has a crucial impact on the firm size distribution, it is natural to ask if $\alpha$ has an impact on the efficiency of the matching procedure. We define the efficiency of matching as the sum of all the surpluses of matched pairs. This surplus is maximal when all agents are matched to their preferred firm.

From Equation \eqref{hxy}, we can compute the total efficiency of the matching:
\begin{align*}
    \text{Eff} (\alpha)&=\int h_\alpha(x,y)\sigma(x,y) dx dy \\
    &= \int \frac{\mu K(\mu + \lambda_{tot})}{K\lambda_{tot}+\mu}\frac{\ell(x) \lambda_{\alpha}(y)}{G_\alpha(x,y)^2}\sigma(x,y) dx dy \\
    &\propto \int \frac{\ell(x) (\alpha s_\alpha(y)+(1-\alpha))}{\left( r_f + \int \mathbbm{1}[\sigma(x,y')>\sigma(x,y)] (\alpha s_\alpha(y')+(1-\alpha)) dy' \right)^2}\sigma(x,y) dx dy.
\end{align*} 
where $s_{\alpha}(y)$ is the market share of firm $y$ when the meeting rate is equal to: \\
$\lambda(y)=\lambda_{tot} \left( \alpha s(y) + 1-\alpha \right)$. \\ 

Taking $\sigma(x,y)=1-d(x,y)$ and using our numerical approximation of $s$, we plot the efficiency of the matching as a function of $\alpha$, for a given level of frictions ($\lambda$, $\mu$ and $K$ are fixed). The graph is presented in Figure \ref{fig:Efficiencyvsalpha}, for $r_f=2$ and for a Gaussian $\ell$. We notice that $\text{Eff} (\alpha)$ is non monotonic in $\alpha$, which means that, for a given level of frictions, there is an optimal $\alpha$ that maximizes the efficiency and that is neither $0$ nor $1$. This implies that the firm size distribution can be either too concentrated or too flat for efficient matching to occur. The optimal value of $\alpha$ changes with the level of frictions, as illustrated by Figure \ref{fig:Efficiencyvsalphab}. \\

Let us assume that one agent wants to maximize her own match surplus, taking the $\tilde \alpha$ chosen by other agents as given, not knowing the position of its preference $x$ in the distribution (but knowing the global shape of the preference distribution $\ell$). Then her optimization program is: 
\begin{align*}
    \alpha(\tilde \alpha)& =\argmax_\alpha \int \frac{\ell(x) (\alpha s_{\tilde \alpha}(y)+(1-\alpha))}{\left( r_f + \int \mathbbm{1}[\sigma(x,y')>\sigma(x,y)] (\alpha s_{\tilde \alpha}(y')+(1-\alpha)) dy' \right)^2}\sigma(x,y) dx dy. \\
    & :=\argmax_\alpha \mathcal{U}_{\tilde \alpha}(\alpha)
\end{align*} 
 The Nash equilibrium is then given by the equation $ \alpha(\tilde \alpha^*)=\tilde\alpha^*$. In general $\tilde\alpha^*\neq \argmax_\alpha \text{Eff} (\alpha)$, as illustrated in Figure \ref{fig:EffUt}, where we plot both $\text{Eff} (\alpha)$ and  $\mathcal{U}_{\tilde \alpha}(\alpha)$. For a single-peaked $\ell$ function for instance, numerical simulations indicate that the best response to any value of $\tilde \alpha$ is always $1$. This does not come as a surprise: it is in the agent's interest to get the most information out of the firm size distribution and to meet first firms $x$ with high values of $\ell (x)$, since they are most likely to be the agent's favorite - because they are the favorite of a large share of agents. This means that a social planner would choose a different $\alpha$ than individual maximizing agents.  \\

This phenomenon comes from the fact that agents do not internalize their effect on the  firm size distribution. They typically want a higher value of $\alpha$ than the optimal one, since they do not take into account that their choice and their match will concentrate the firm size distribution beyond the optimal level. We discuss in Section \ref{sec:heterogmeetingrates} how this phenomenon can relate to the increase in $\alpha$ over the years. \\ 

We have seen in this Section that, in presence of heterogeneous preferences, the impact of search frictions crucially depends on the shape of meeting frictions. In the following section, we study the empirical shape of the meeting rates in the case of Firm to firm export. In particular we ask if modelling meeting rates as constant matches the data. If so, then frictions always have an homogenizing effect and a decrease in friction intensity causes a concentration of the market as stated in \cite{bai2020search,albrecht2022vertical,lenoir2022search,mazet2021information}. If not, then this needs not be the case.

\begin{figure}
\begin{subfigure}[t]{0.45\textwidth}
        \centering\includegraphics[width=\linewidth]{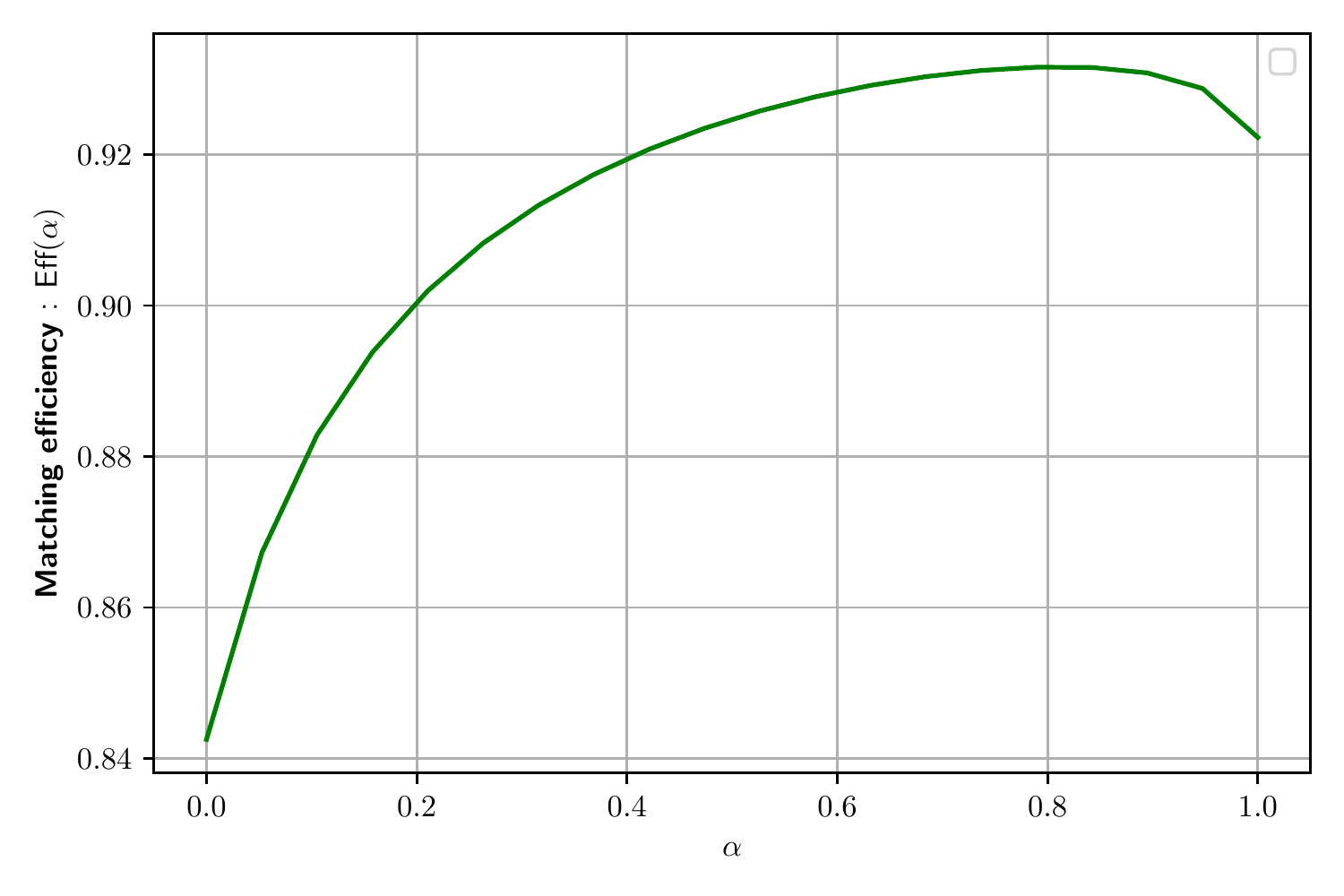}
        \caption{Matching efficiency depending on $ \alpha$, $r_f=0.2$.}
        \label{fig:Efficiencyvsalpha}
\end{subfigure}
\hfill
\begin{subfigure}[t]{0.45\textwidth}
        \centering\includegraphics[width=\linewidth]{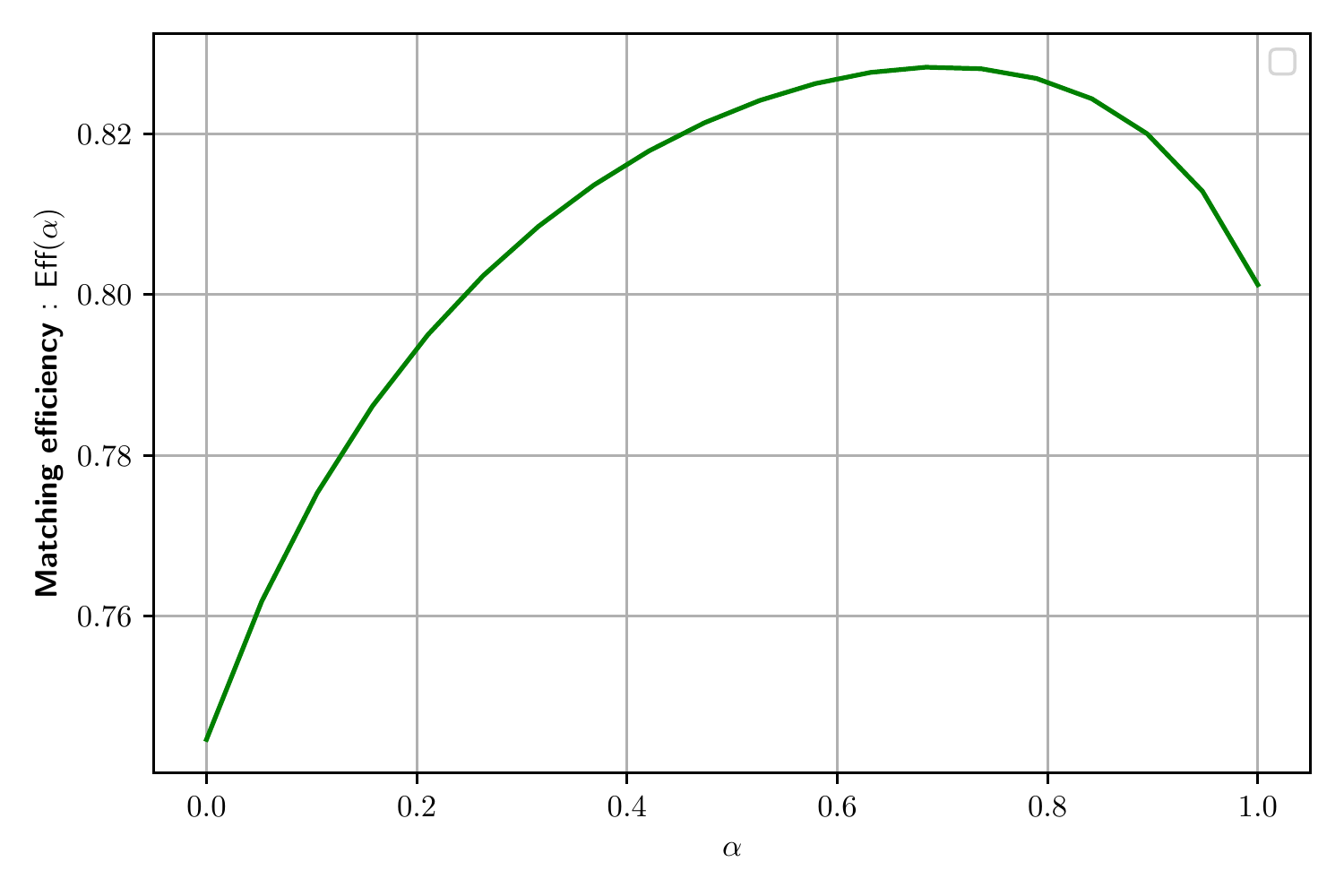}
        \caption{Matching efficiency depending on $ \alpha$, $r_f=0.5$.}
        \label{fig:Efficiencyvsalphab}
\end{subfigure}        
\end{figure}

\begin{figure}
    \centering
    \includegraphics[width=0.75\linewidth]{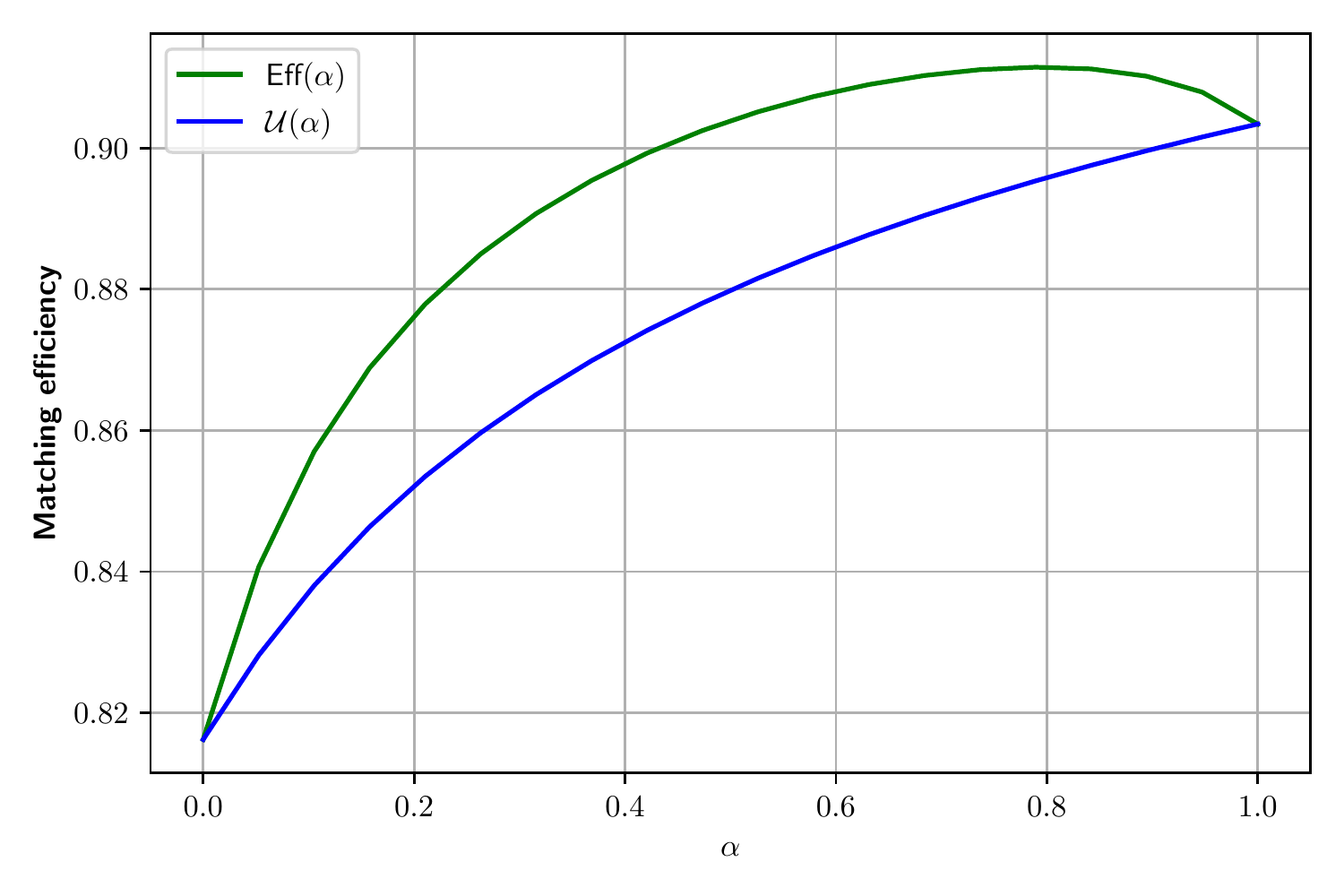}
    \caption{Matching efficiency and utility depending on $ \alpha$, for $\tilde \alpha$=1.}
    \label{fig:EffUt}
\end{figure}

\section{A case for heterogeneous meeting rates}
\label{sec:heterogmeetingrates}

In the previous Sections, we have shown that, in presence of heterogeneous agents preferences, the way the meeting rates $\lambda(y)$ depend on the firms' market shares $s(y)$ is decisive for the equilibrium firm size distribution. In particular, the parameter $\alpha$, which governs the dependency between the meeting rates and the market shares, determines whether frictions tend to favor the least preferred or the most preferred firms of the market. 
In this Section, we evaluate the value of $\alpha$, to determine which of these regimes is closest to reality. We first review what the previous literature has shown. Second, we estimate $\alpha$ for the international goods market, where foreign buyers match with French firms. We find that the meeting rate is close to proportional to the firms' market share, and that $\alpha$ has increased over the years. \\ 

\subsection{Literature}
In Labor economics, two extreme benchmark cases are often analyzed: random matching in which all firms have an equal probability of being sampled - this corresponds to the constant meeting rate case $\alpha=0$ - and balanced matching, in which the probability of being sampled is proportional to firm size - this corresponds to the case where the meeting rate is proportional to the market share $\alpha=1$.
A few papers have investigated the theoretical implications of having either random or balanced matching. For instance, \cite{burdett1988balanced} use a wage-posting model, and study specifically the implications of balanced matching, with respect to random matching on the wage distribution. \cite{mortensen1994personal} analyze a form of matching that mixes random matching and balanced matching. 
\cite{postel2002equilibrium} write and estimate a sequential auctions model, which does not specify any a priori functional form for the meeting rates. One theoretical result is that with balanced matching, the most productive firm hires all the workers of the economy. 
The estimation of the model leads them to reject both the assumptions of random and balanced matching.\footnote{More precisely, they find that firm's size is not monotonic in the hiring effort, and conclude that both the assumptions of random and balanced matching are rejected by the data.} To the best of our knowledge, the precise estimation of the way meeting rates depend on the size of firms has not been undertaken before in the labor literature.

The International Trade literature has adapted the Labor literature to study the matching of buyers to suppliers.
\cite{lenoir2022search} assume in their baseline model that the meeting rate of buyers is constant across suppliers, which corresponds to $\alpha=0$ in our setting. Then, they extend the model to a case where the meeting rate increases with the supplier's productivity. They show that frictions have qualitatively the same distortive impact in this extended setting. Yet, they do not estimate the shape of the meeting rate. 
\cite{eaton2019firm} also allow for the meeting rate with suppliers to depend on the supplier's productivity. They model congestion effects on the sellers' side: the matching intensity for a given supplier decreases with the mass of sellers which are more productive. Therefore the more productive a seller, the higher the meeting rate of a given buyer with this seller. 
As their paper estimates this congestion parameter assuming complex parametric functional forms and using the hypothesis of homogeneous preferences, it is hard to relate this parameter to $\alpha$ and to our work.

\subsection{Empirical evidence}

\textbf{Theory} Using the notations of Section~\ref{sec:multifirms}, we formulate the flows $f_t(u,y)$ of unmatched agents into the firm indexed by $y$ at time $t$ as follows:
\begin{equation*}
    f_t(u,y) = \lambda_{0t}(y) \times U_t
\end{equation*}
where $U_t$ is the number of unmatched agents at time $t$, and $\lambda_{0t}(y)$ is the meeting rate, by unmatched agents, of a firm of type $y$ at time $t$. Recall that we assumed that the meeting rate of firm $y$ takes the following linear form: 
\begin{equation*}
\lambda_{0t}(y) = \left( \frac{1-\alpha_t}{N} + \alpha_t s_t(y) \right) \lambda_{tot,t}^0
\end{equation*}
with $N$ the number of firms in the market, $s_t(y)$ the market share of firm $y$ among matched agents, and, by definition,  $\lambda_{tot,t}^0 = \sum_y \lambda_{0t}(y) $. \\

From these equations, we can write the share $s_t(u,y)$ of flows from the unmatched state that each firm $y$ gets at time $t$ as: 
\begin{equation*}
    s_t(u,y) = \frac{f_t(u,y)}{\sum_{y' \in Y} f_t(u,y')} = \frac{1-\alpha_t}{N} + \alpha_t s_t(y)
\end{equation*} 
Regressing  $s_t(u,y)$ on the market share $s_t(y)$ enables to recover an estimate of $\alpha_t$. Note that our estimation strategy does not rely on using steady-state outcomes given by our model. Our model is applicable to any setting with many-to-one frictional matching.\footnote{More precisely, our model applies to any partial equilibrium model with many-to-one matching, with no capacity constraint on the firm side, transferable utility and on-the-job search.}  We choose to apply it the international goods market.  \\ 

\textbf{Matching on the international goods market} We exploit a French firm-to-firm trade dataset which records the identity, for each French exporter, of its non-domestic buyers in the European Union. 
The dataset records for each transaction the identity of the exporting firm (its SIREN identifier), the identification number of the importer (an anonymized version of its VAT code), the date of the transaction (month and year), the product category (at the 8-digit level of the combined nomenclature) and the value of the shipment. The data cleaning and product harmonization over time for this data are described in \cite{bergounhon2018guideline}. We perform the analysis over the period 1996-2016. We define a buyer as a buyer $\times$ 8-digit product combination, and select each year the buyers which import from at most one French firm this year and the year before.\footnote{The 8-digit products are harmonized over the period 2005-2017, as explained in \cite{bergounhon2018guideline}.} This eliminates 6.4\% of buyers. Our final dataset includes 20,925,116 foreign buyers and 92,310 French firms. We then define markets as year $\times$ country $\times$ 6-digit product cells. One foreign buyer can thus belong to only one market. We have 39,549 markets in total.\footnote{In what follows, we are able to estimate $\alpha$ for 22,969 markets.} Note that declaration rules for exports change over our period of analysis.\footnote{From 1993 to 2000, the declaration threshold was set at 250,000 euros of exports in a given year. In 2001, it was set at 650,000. From 2002 to 2006, it was set at 100,000. From 2007 to 2010, it was set at 150,000. From 2011 to 2018, it was set at 460,000. See \cite{bergounhon2018guideline}.} We discuss this issue in the robustness tests below.  \\

First, we compute for each year $t$, each French firm and each market the number of buyers which were not matched in year $t-1$ to a French supplier and to which the firm starts selling in year $t$. Note that this measure is not exactly the empirical counterpart of $f_t(u,y)$. In our dataset, we only observe the connections of foreign buyers with French firms. We therefore do not know whether a foreign buyer not in our dataset is matched with a non-French firm or whether it is unmatched. What we observe is $f_t(u,y)+f_t(\bar{F},y)$, i.e. the flows into a supplier $y$ of buyers coming from non-French suppliers and the flows into a supplier $y$ coming from the unmatched state, but we cannot disentangle those two elements with our dataset.

To deal with that issue, 
we assume that among the flow into $y$ of buyers that were already matched, a fraction $\beta_1$ of them were matched with a French supplier.
In other words, we assume that there exists $\beta_1 \in [0,1]$ such that for each agent $x$ and each firm $y$:
\begin{align*}
    &\underbrace{\int_{y' \in F} \mathbbm{1}[\sigma(x,y)>\sigma(x,y')] h(x,y') dy'}_{f(F,y)}  = \\
    &\beta_1 \left( \underbrace{\int_{y' \in F} \mathbbm{1}[\sigma(x,y)>\sigma(x,y')] h(x,y') dy'}_{f(F,y)} + \underbrace{\int_{y' \in \bar{F}} \mathbbm{1}[\sigma(x,y)>\sigma(x,y')] h(x,y') dy'}_{f(\bar{F},y)} \right)
\end{align*}
 $\beta_1$ is the share of buyers who, conditionally on arriving to a firm $y$ and on being previously matched with a supplier in the EU, were previously matched with a French supplier.
 An approximation for $\beta_1$ is the share of France in the European Union's GDP. We use $\beta_1 = 0.15$ for our estimation, but we show that our estimated $\hat{\alpha}$ is not very sensitive to $\beta_1$.  \\
 
 We can therefore reconstruct from the data the flows of buyers from the unmatched state into firm $y$ as follows:
\begin{equation}
    f_t(u,y) = \underbrace{f_t(u,y)+f_t(\bar{F},y)}_{observed} - \frac{1-\beta_1}{\beta_1} \underbrace{ f_t(F,y)}_{observed}
\end{equation}\label{eq:adjustment_flows}
Second, we compute the total number of previously-unmatched buyers, the empirical counterpart of $\sum_y f_t(u,y)$,  to obtain the share $s_t(u,y)$. Third, we compute the market share $s_t(y)$ of each French firm in each market: we divide the number of buyers a firm has in a market in a given year by the total number of buyers of this market importing from France this year. \\

\textbf{Results} Figure~\ref{fig:f_totnbnewbuyers_s_y} displays the correlations between the share  $s(u,y)$ of the flows of unmatched buyers with the market share $s(y)$ of firms. 
The relationship seems fairly linear, which validates the functional form we used in our theory for the meeting rate. We then perform the regression of $s(u,y)$ on $s(y)$ market by market, in order to estimate distinct values of $\alpha$ for each market. We find that the average $\hat{\alpha}$ is equal to $0.75$. 
All standard test reject at any level the hypothesis that $\alpha=0$, hypothesis under which most of the literature operates as we saw in the previous section. 
With this estimation of $\alpha$ in hand, we show in the Appendix, see Figure~\ref{fig:finalsimulation}, that with plausible values of search frictions and preferences' heterogeneity, a rise in frictions leads to a higher concentration of the market, rather than a lower concentration of the market. The rise in market concentration in the last decades thus cannot plausibly be explained by the decrease in search frictions. \\

In Figure~\ref{fig:alpha_year}, we show that the estimated $\alpha$ has risen over time. It went from 0.63 in 1996 to 0.84 in 2016. This is one of the main result of our study. This rise can be a novel explanation for the rise in market concentration over the years. Indeed, we showed in Section \ref{sec:multifirms} that a rise in $\alpha$ causes in most cases a concentration of the market for a fixed friction intensity. Therefore, Figure~\ref{fig:alpha_year} proposes a new diagnosis on the market concentration: it might - at least partially - be due to a change in the structure of frictions, and not only by a reduction of friction intensity. \\

This rise in the dependency of the meeting rate with respect to the size of the firm could be explained, among other, by the rise of digital platforms. Indeed, in a world where suppliers or jobs are mainly found locally and through word of mouth, agents first meet firms that are geographically close. If agents and firms are distributed the same way across the territory, agents meet all firms at the same rate, regardless of their size. On the contrary, digital platforms allow agent to meet bigger firms first, \emph{because it is what the agents find optimal for themselves}, as mentioned in Section \ref{sec:multifirms}. It is therefore possible that digital platforms have two different effects that can both lead to concentration and do not have the same consequence in terms of efficiency: digital platforms can allow agents to meet more firms - reducing the intensity of friction $r_f$ - and also allow agents to meet \emph{bigger firms} - increasing the value of $\alpha$.  \\

Note that the shape of the preferences' distribution is likely to vary a lot across markets. Such an estimation is beyond the scope of the paper, and is left for future work. We devise, in a companion paper, a methodology to estimate jointly the heterogeneity of preferences and the search frictions, and apply it to the international goods market.

\begin{figure}
    \centering
    \includegraphics[width=0.6\linewidth]{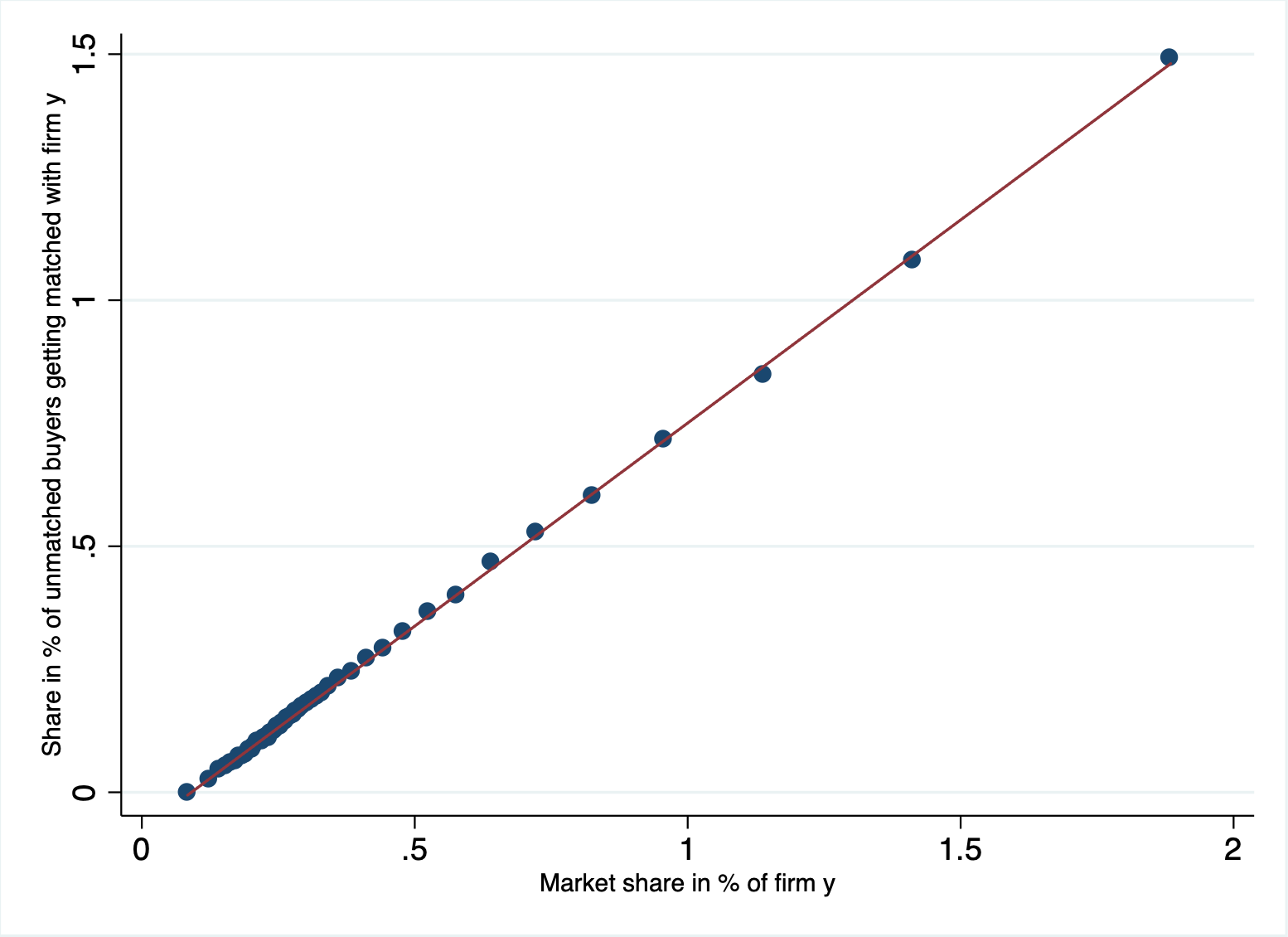}
    \caption{Correlation between $s(u,y)$ and $s(y)$}
    \label{fig:f_totnbnewbuyers_s_y}
    \parbox{1\linewidth}{\footnotesize Note: the market share of each firm is determined for each market, where a market is a EU country x 6-digit product x year cell. We include market fixed effects in the regression.}
\end{figure}

\begin{figure}
    \centering\includegraphics[width=0.7\linewidth]{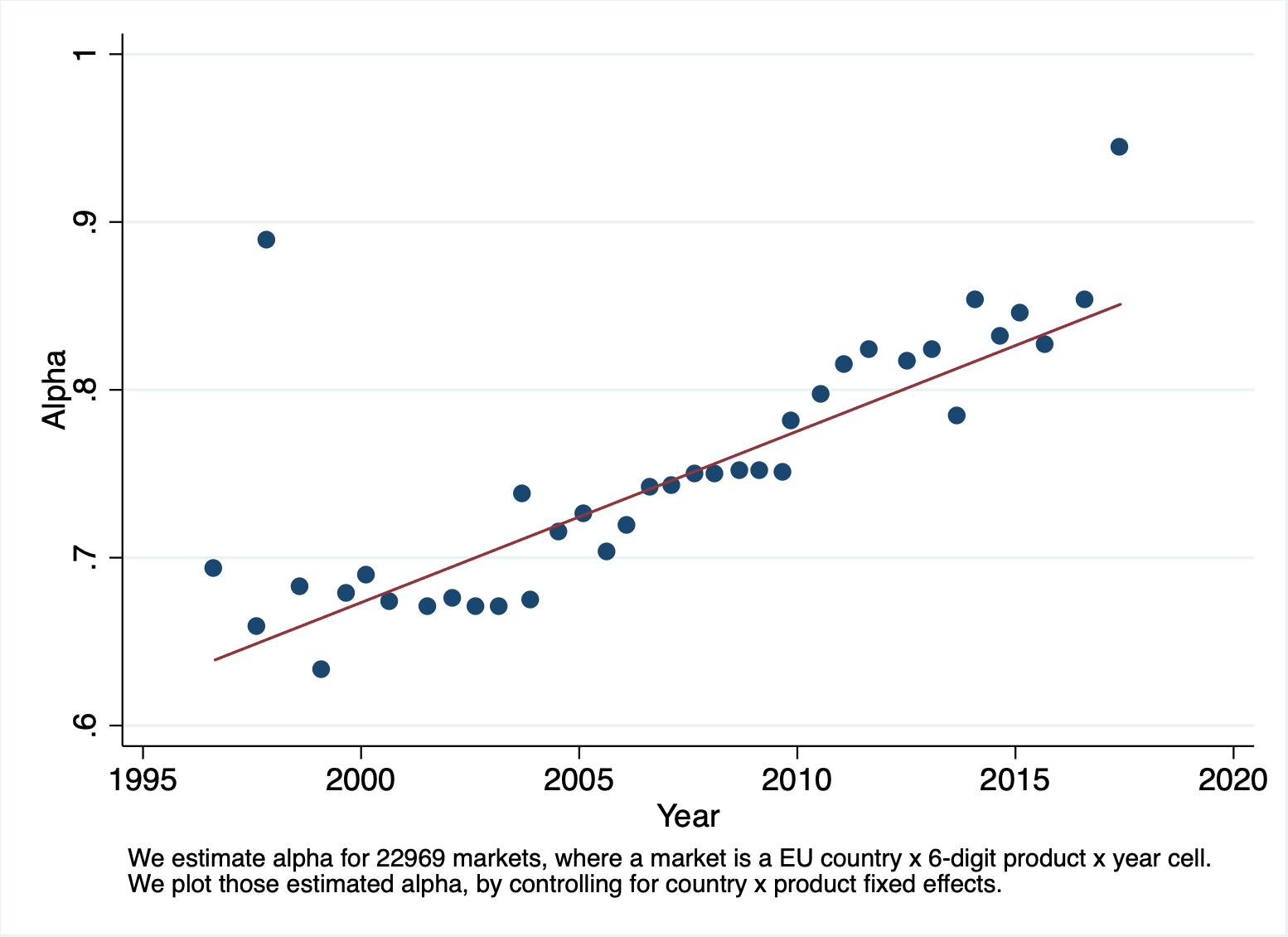}
    \caption{The evolution of the structure of meeting rates over the years}
    \label{fig:alpha_year}
\end{figure}
\vspace{0.5cm}

\subsection{Robustness tests}

\textbf{Declaration rules} Declaration rules for exports change over our period of analysis. In particular, they become more stringent in 2011, when the threshold is set at 460,000 euros of exports per year per firm. We perform here two exercises to test the robustness of our results to these changing thresholds. First, we assess whether the evolution of the structure of meeting rates over the year, exhibited in Figure~\ref{fig:alpha_year}, comes partly from the declaration thresholds. We check this is not the case by estimating for each market and year the $\alpha_t$, conditional on French exporters selling more than 650,000 euros abroad. This thresholds corresponds to the most stringent threshold over the period - set in 2001 - and therefore the $\alpha_t$ are more directly comparable across years. Second, our baseline estimation might be an over-estimation of $\alpha_t$, due to non-zero declaration thresholds.

\section{Conclusion}
\label{sec:conclusion}

In this paper, we study the effect of search frictions on the firm size distribution in presence of heterogeneous preferences. The literature has so far analyzed the role of search frictions in presence of homogenous agents' preferences: it was assumed that all workers for instance prefer the same employer. In this context, the higher the search frictions, the less concentrated the firm size distribution: search frictions advantage the least-productive firms in the market.
With heterogeneous preferences, we assume there is no longer a common ranking of firms by agents, so that the optimal firm size distribution corresponds to the agents' preferences distribution. 
In this context, frictions induce deviations between the optimal and the realized size distribution, but they can now do so in diverse ways: they can have either an homogenizing effect or a concentrating effect. 

In this paper, we provide a theoretical framework, inspired by \cite{lindenlaub2016multidimensional}, to study this question. We develop a partial equilibrium model of the matching process between atomistic agents and firms, which can be applicable to many settings, such as the labor markets or the goods markets.

We show that frictions have an homogenizing effect if meeting rates are constant across firms. This effect is the same as in the homogenous preferences case. 
However, when the meeting rate is proportional to the market share, higher frictions may lead to a more concentrated market, as compared to agents' preferences. In other words, frictions may now advantage the most preferred firm in the market.

The shape of meeting frictions is therefore crucial to understand the impact of frictions on the firm size distribution. We explore empirically which shape of the meeting rate is the most likely to prevail. We study the international goods market - and show that the meeting rate is close to proportional to the firms' market share. We show that with plausible values of search frictions and preferences’ heterogeneity, a rise in frictions leads to a higher concentration of the market, rather than a lower concentration of the market.

Our results have important policy implications.
First, our results have implications for the redistributive effects of some policies, especially policies aiming at reducing frictions on the labor or the good markets. 
Second, the change in the structure of search frictions can be one plausible explanation for the rise in market concentration, occurring in the last decades. It is commonly thought that search frictions have decreased in the last decades, and that could explain partially the rise in concentration. We show that the structure of search frictions might matter as much as the intensity of search frictions. 

\newpage
\bibliographystyle{apalike}
\bibliography{biblio.bib}

\clearpage

\appendix

\clearpage

\section{Annex to Section \ref{sec:twosup}}

Theorem \ref{alwaysloser} characterizes the cases where the effect of friction is ambiguous.

\begin{theorem*} 
When $r_f \geq \frac{\alpha-1/2}{1-\alpha} $, then $s_A \leq p_a$ for all $p_a \geq 0.5 $ (frictions have a homogenizing effect). \\
When $r_f < \frac{\alpha-1/2}{1-\alpha} $, then $s_A > p_a$ for at least some $p_a \geq  0.5 $ (frictions may accentuate existing inequalities).
\end{theorem*}

\begin{proof}

From Equation \eqref{equationsa}, we can write, with $\lambda_A(s_A)= \left( (1-\alpha)/2 + \alpha s_A \right) \lambda_{tot}$ and $\lambda_B(s_A)= \left( (1-\alpha)/2 + \alpha(1- s_A) \right) \lambda_{tot}$, 

\begin{align*}
    p_a(s_A) &=\lambda_{tot}  (\mu+\lambda_A(s_A)) (\mu+\lambda_B(s_A))\frac{s_A- \left( \frac{\lambda_A(s_A)}{\mu+\lambda_B(s_A)} \times \frac{\mu}{\lambda_{tot}} \right) }{\lambda_A(s_A) \lambda_B(s_A) (2\mu+\lambda_{tot})} \\ 
 &=\lambda_{tot}  (\mu+\lambda_A(s_A)) \frac{(\mu+\lambda_B(s_A))s_A- \left( \lambda_A(s_A) \times \frac{\mu}{\lambda_{tot}} \right) }{\lambda_A(s_A) \lambda_B(s_A) (2\mu+\lambda_{tot})}.
\end{align*}

We denote $f(s_A)= p_a(s_A)-s_A$. We note that $f(x)=P(x)/Q(x)$ where $P$ is a third degree polynomial with a positive leading coefficient and $Q$ is a second degree polynomial with a negative leading coefficient.
Because $A$ and $B$ are symmetric, we know that $f(1-x)=-f(x)$ and we notice that $Q(1-x)=Q(x)$, so $P(1-x)=-P(x)$. As $P$ is anti-symmetric about $1/2$ and a third degree polynomial, we know its derivative reach a minimum in $1/2$. Therefore, if its derivative in $1/2$ is positive, then its derivative is positive in all points and then $P(x) >0 $ for all $x > 1/2$. However if $P'(1/2)<0$, as $Q'(1/2) =0$, then $f'(1/2)<0$ and $f(x) < 0$ for some $x > 1/2$, meaning that $s_A>p_a$ for some $s_A >1/2$. \\

The only thing left is to compute the value of $P'(1/2)$. Straightforward computations give:
$$ P'(1/2)/\lambda_{tot}^3  =  r_f^2-\alpha r_f^2-\alpha r_f+r_f/2.$$

This in turn, leads to the result of Theorem \ref{alwaysloser}.
\end{proof}

\begin{theorem*} 
When $p_a \geq (\frac{1-\alpha}{2}+\alpha p_a)(\frac{1-\alpha}{2}+ p_a+ 2 \alpha p_a(1-p_a)) $, and $p_a >0.5$ then $s_A \leq p_a$ for all $p_a \geq 0.5 $ (frictions have a homogenizing effect). \\

\end{theorem*}

\section{Annex to Section \ref{sec:multifirms}}
\label{annex_sec_multifirms}

Let us denote, as in \cite{lindenlaub2016multidimensional} $F_{\sigma \mid x}$  the conditional sampling cdf of value $\sigma$, given $x$ (with density $f_{\sigma \mid x}$) so that 
$$\bar{F}_{\sigma \mid x}(s)=\mathbb{E}_y\left[\mathbf{1}\left\{\sigma\left(x, y\right)\leq s\right\}\right].$$

We also consider the following functions

$${H}_{\sigma \mid x}(s)=\mathbb{E}_y\left[h(x,y)\mathbf{1}\left\{\sigma\left(x, y\right)<s\right\}\right], $$
with density $\chi_x(s)$. We note that $\chi_x(s)= \sum_{y| \sigma(x,y)=s} h(x,y)\times \left(\frac{\partial\sigma(x,y)}{\partial y}\right)^{-1}$. We also introduce
$${\Lambda}_{\sigma \mid x}(s)=\mathbb{E}_y\left[\lambda_1(y)\mathbf{1}\left\{\sigma\left(x, y\right)<s\right\}\right], $$
with density $\lambda_x(s)=\sum_{y| \sigma(x,y)=s} \lambda_1(y)\times \left(\frac{\partial\sigma(x,y)}{\partial y}\right)^{-1}$.

 Equation \eqref{eq:master} can be written as
 
 \begin{align*} 
h(x,y) \times \left[\mu+\bar{\Lambda}_{\sigma \mid x}(\sigma(x,y)) \right] =  
\lambda_0(y) u(x) + \lambda_1(y) {H}_{\sigma \mid x}(\sigma(x,y))
\end{align*}
Summing over $y$ such that $\sigma(x,y)=s$ leads to:
\begin{equation}\label{eq:formeS}
   \left[\mu+\bar{\Lambda}_{\sigma \mid x}(s) \right] \chi_x(s) =K\lambda_x(s)   u(x)+\lambda_x(s)  \int_0^s \chi_x(s') d s' . 
\end{equation}

We denote $G(s)=\left[\mu+\bar{\Lambda}_{\sigma \mid x}(s) \right]$. Dividing by $ \lambda_x(s)$ and differentiating in $s$ we get:

$$
\frac{G(s)}{\lambda(s)} \chi'_x(s)-\frac{\lambda_x(s)^2 +G(s)\lambda_x'(s)}{\lambda_x(s)^2} \chi_x(s) =   \chi_x(s) .
$$

Therefore
$$
 \frac{\chi'_x( s)}{\chi_x( s)} = -2 \frac{G'(s)}{G(s)}+ \frac{\lambda_x'(s)}{\lambda_x(s)}.
$$

This ODE solves as:
$$
\chi_x( s)= C_x\frac{\lambda_x(s)}{G(s)^2},
$$
where $C_x$  does not depend on $s$. We find $C_x$ by plugging this form in Equation \eqref{eq:formeS}, leading to 
$$ \frac{\lambda_x(s)}{G(s)}=Ku(x)+\int_0^s\chi_x(s') ds'.$$
Taking $s=0$ gives $C_x= Ku(x)(\mu+\lambda_{tot})$.

We can write this as $$\sum_{y, \sigma(x,y)=s} h(x,y)\left(\frac{\partial\sigma(x,y)}{\partial y}\right)^{-1}= \frac{C_x}{G(s)^2}\sum_{y, \sigma(x,y)=s} \lambda_1(y)\left(\frac{\partial\sigma(x,y)}{\partial y}\right)^{-1} $$

As we know from Equation \eqref{eq:master} that among the different $y$ such that $\sigma(x,y)=s$, $h(x,y)/\lambda_1(y)$ is constant, we conclude that 
$$h(x,y)= C_x\frac{\lambda_1(y)}{G(s)^2}.$$ 
In conclusion
$$h(x,y)=\frac{Ku(x)(\mu+\lambda_{tot})\lambda_1(y)}{G(s)^2}.$$

\section{Annex to Section \ref{sec:multifirms} part 2}
\label{annex_sec_multifirms_part2}

First we look at the case where $r_f$ goes to $0$. We use the fact that for $s(y) >0$

$$s(y)= \frac{\lambda(y)}{\lambda_{tot}}\int \frac{r_f (r_f + 1)\ell(x)}{(r_f+\int_{y'}\frac{\lambda(y')}{\lambda_{tot}}\mathbbm{1}[\sigma(x,y')>\sigma(x,y)] dy')^2} dx$$

We know that, in that set up, 
$$\tilde G(x,y)=(r_f+\int_{y'}\frac{\lambda(y')}{\lambda_{tot}}\mathbbm{1}[\sigma(x,y')>\sigma(x,y)] dy')^2 =(r_f+ 2\frac{\lambda(y)}{\lambda_{tot}}|y-x|+o(|y-x|))^2 $$ 
when $x$ goes to $y$ and that  $\tilde G(x,y)$ goes to a non-zero value otherwise. We can integrate between $y-\eta$ and $ y+\eta$ as the rest of the integral goes to zero when $r_f$ goes to $0$. Using a change of variable in the integral, we see that 

$$\int_{y-\eta}^{y+\eta} \frac{r_f (r_f + 1)\ell(x)}{(r_f+ 2\frac{\lambda(y)}{\lambda_{tot}}|y-x|)^2 } dx \rightarrow \frac{\ell(y)\lambda_{tot}}{\lambda(y)},$$

which gives the result $s(y)\rightarrow  \ell(y)$.

For the case where $r_f$ goes to infinity, we work in analogy with Section \ref{sec:twosup}.

Let us denote $$ \mathcal{S}=\left\{[x, x+1/2]\ | \ x \in [0, 1/2]\right\}\bigcup \left\{[0, x]\cup[1/2+x, 1]\ | \ x \in [0, 1/2]\right\}$$ the set of intervals of size $1/2$ on the interval $[0,1]$ regarded as a circle and $\mathcal{S}_y=\{I \in | y \notin I\}$ the set of intervals in $\mathcal{S} $ that does not contain $y$.
Let us take $y^*$ such that $\int_I \ell(x) < 1/2$ for every interval $I$ in $ \mathcal{S}_{y^*}$ . Then let us consider $I_\epsilon= [y^*-\epsilon, y^*+\epsilon] $ and $I_{-\epsilon}=[0,1] \backslash I_\epsilon$ for some small $\epsilon$.

When $r_f$ goes to infinity, the number of agents that have meet $n$ suppliers if of order $\mathcal{O}(r_f^{-n})$. We keep only the first two orders for matched agents: we assume that all matched agents have encountered either $1$ (first order) or $2$ (second order) suppliers. 

Then, we group together all firms in $I_\epsilon$ and see them as a ``super-firm'' denoted $A$ and all firms in $I_{-\epsilon}$ as a ``super-firm'' denoted $B$. This analogy only works in regime where $r_f$ goes to infinity because we exclude agents that would go from $B$ to $A$ and back to $B$ as this is a third order fraction of all agents.

Following notations from Section \ref{sec:twosup}, we denote $s_A= \int_{I_\epsilon} s(y) dy$ and $\lambda_A$, $\lambda_B$ accordingly. We now from Section \ref{sec:twosup} that $s_A$ either goes to one or to zero depending on $p_a$ the proportion of agents that prefer firm $A$. Here the first order (in $\epsilon$) approximation of $p_a$ is given by: 

$$\int \int_{y' \in I_{-\epsilon}} \mathbbm{1}[\sigma(x,y')<\sigma(x,y^*)] \frac{\lambda(y')\ell(x)}{\lambda_B} dy' dx,$$

approximation made assuming all $\sigma(x,y)\simeq \sigma(x,y^*) $ for $y \in I_\epsilon$. We give a lower bound on that quantity by giving an upper bound on $p_b=1-p_a$:
\begin{align*}
    p_b &\leq \sup_{ \lambda(y')| \int \lambda(y')  = \lambda_B}\int_{y' \in I_{-\epsilon}}\int  \mathbbm{1}[\sigma(x,y')>\sigma(x,y^*)] \frac{\lambda(y')\ell(x)}{\lambda_B} dx dy' \\
    & = \sup_{ y' \in   I_\epsilon} \int  \mathbbm{1}[\sigma(x,y')>\sigma(x,y^*)] \ell(x) dx
\end{align*} 

Taking $\sigma(x,y)= f(d(x,y))$, the set  $\{x |\sigma(x,y')>\sigma(x,y^*) \}$ is a set of $S_{y^*}$. By hypothesis, we have 
$$ \sup_{ y' \in   I_\epsilon} \int  \mathbbm{1}[\sigma(x,y')>\sigma(x,y^*)] \ell(x) dx  < 1/2, $$
so $p_a >1/2$ and, when $r_f$ goes to infinity, $s_B $ goes to 0.

\section{Annex to Section \ref{subsec:numerical}}\label{annex:numerical}

 We now that $s$ is the solution to the following equation:

$$s(y)=(\alpha s(y)+(1-\alpha) )\int \frac{r_f (r_f + 1)\ell(x)}{(r_f+\int_{y'}s(y')\mathbbm{1}[\sigma(x,y')>\sigma(x,y)] dy')^2} dx.$$

We approximate it using the  fixed-point iteration method. We compute $s_{t+1}=F(s_t)$ for the following functional $F$:

$$s_{t+1}(y)=(\alpha s_t(y)+(1-\alpha) )\int \frac{r_f (r_f + 1)\ell(x)}{(r_f+\int_{y'}s_t(y')\mathbbm{1}[\sigma(x,y')>\sigma(x,y)] dy')^2} dx.$$

This sequence can only converge to a fixed-point of $F$, thus it can only converge to $s$.

\section{Annex to Section~\ref{sec:heterogmeetingrates}}

In Section~\ref{sec:heterogmeetingrates}, we provide novel empirical evidence on the shape of meeting frictions for the international goods market. We complement this estimation with a similar estimation for the French labor market. 

\textbf{Empirical methodology} We use matched a French employer-employee dataset - \textit{DADS Postes} - which describes every firm for which every French employee worked during the calendar year, the number of hours worked, and the wage and occupation. Workers have individual identifiers, which allows us to track them across firms over time, and crucially identify flows into each firm of previously unmatched workers. We perform the analysis for the year 2015. We select workers who are working in a given year but were not the year before. 
For each firm, we compute the share of those previously-unmatched workers getting matched to the firm, the empirical counterpart to $s(u,y)$, and regress it on the share of workers that this firm has in 2014, the empirical counterpart to $s(y)$. Our final dataset includes 21.5 million workers and 1,340,482 French firms. \\

\textbf{Results} First, the relationship seems, as for the international goods market, quite linear. Second, the slope of the line gives us: $\hat{\alpha}=0.957$. The estimated $\alpha$ is therefore quite close to what we estimated for the international goods market.  \\ 


\section{Plausible scenarii for the effect of frictions}

 We perform simulations to study the likely impact of search frictions with our estimated values of $\alpha$, around 0.75. The effect of frictions on the firm size distribution depends on $\alpha$, but also on the range of likely values for frictions, and on the heterogeneity of preferences. \par 
Search frictions in the labor market are usually found to be between 0.1 and 0.2 \citep{postel2002equilibrium}.\footnote{\cite{postel2002equilibrium} estimates that $\kappa_1$, which is the meeting rate $\lambda_1$ divided by the exogenous destruction rate $\delta+\mu$ is equal to: 5.84 for unskilled manual workers, 7.07 for skilled manual workers, 7.61 for executives, managers and engineers, 8.87 for technical supervisors and technicians. This yields, in our notation, an approximate of the intensity of search frictions $\frac{\mu}{\lambda_{tot}}$ between 0.11 and 0.17.}
Based on such values, we compute the equilibrium firm size distribution for frictions increasing from 0 to 0.3. One aspect we do not have much information on is how heterogeneous preferences are. 
In Figure~\ref{fig:syVSyPROP3c}, we show that even with a mildly-heterogeneous preferences' distribution, higher frictions concentrate the distribution around the median preference - 0.5 in our example. Figure~\ref{fig:syVSyPROP3d} displays a similar pattern with stronger heterogeneous preferences. \par 

\begin{figure}
\centering
\begin{subfigure}[t]{0.45\textwidth}
        \includegraphics[width=\linewidth]{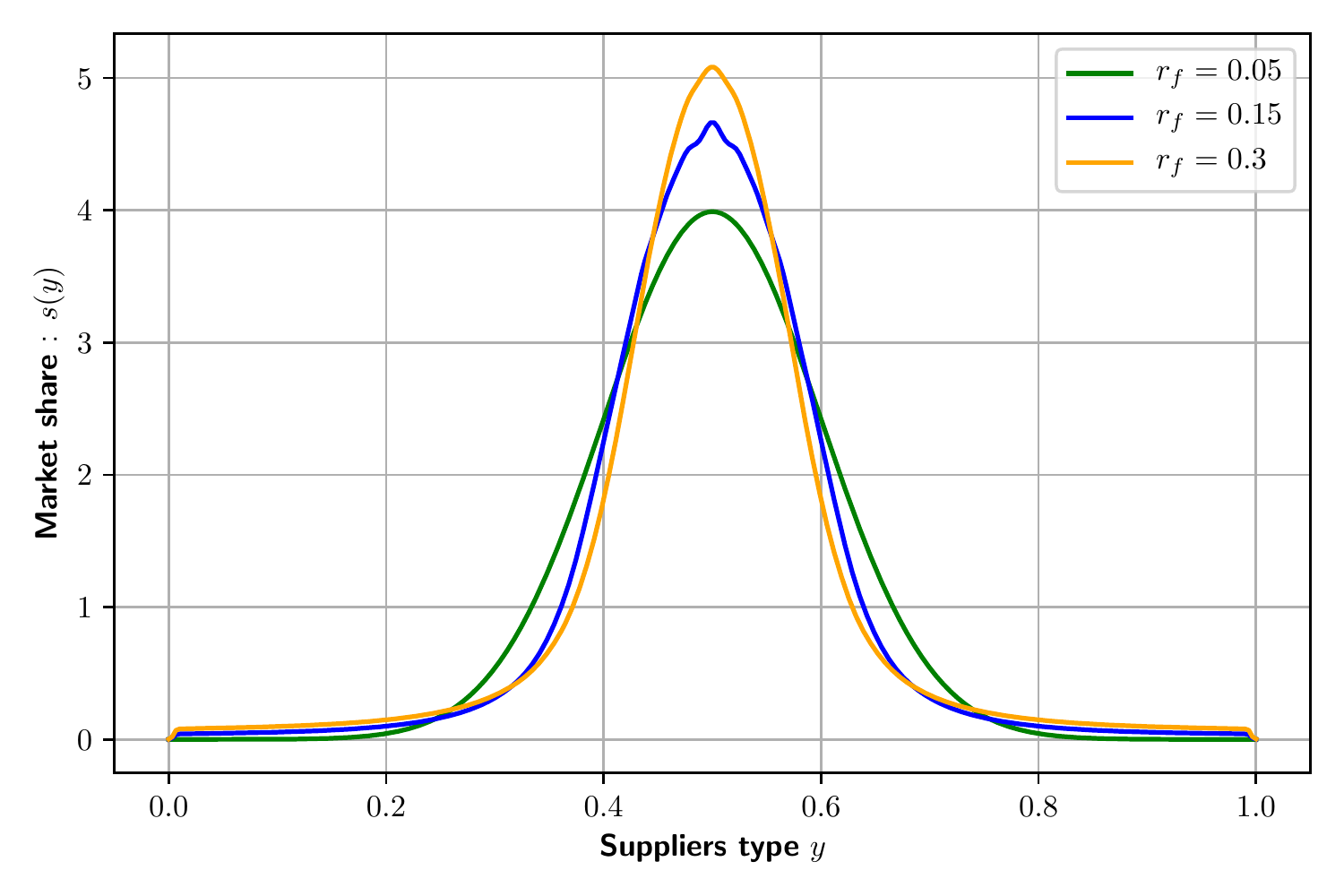}
        \caption{Mildly-heterogenous preferences, with normal-shaped preferences}
        \label{fig:syVSyPROP3c}
\end{subfigure}
\hfill
\begin{subfigure}[t]{0.45\textwidth}
        \includegraphics[width=\linewidth]{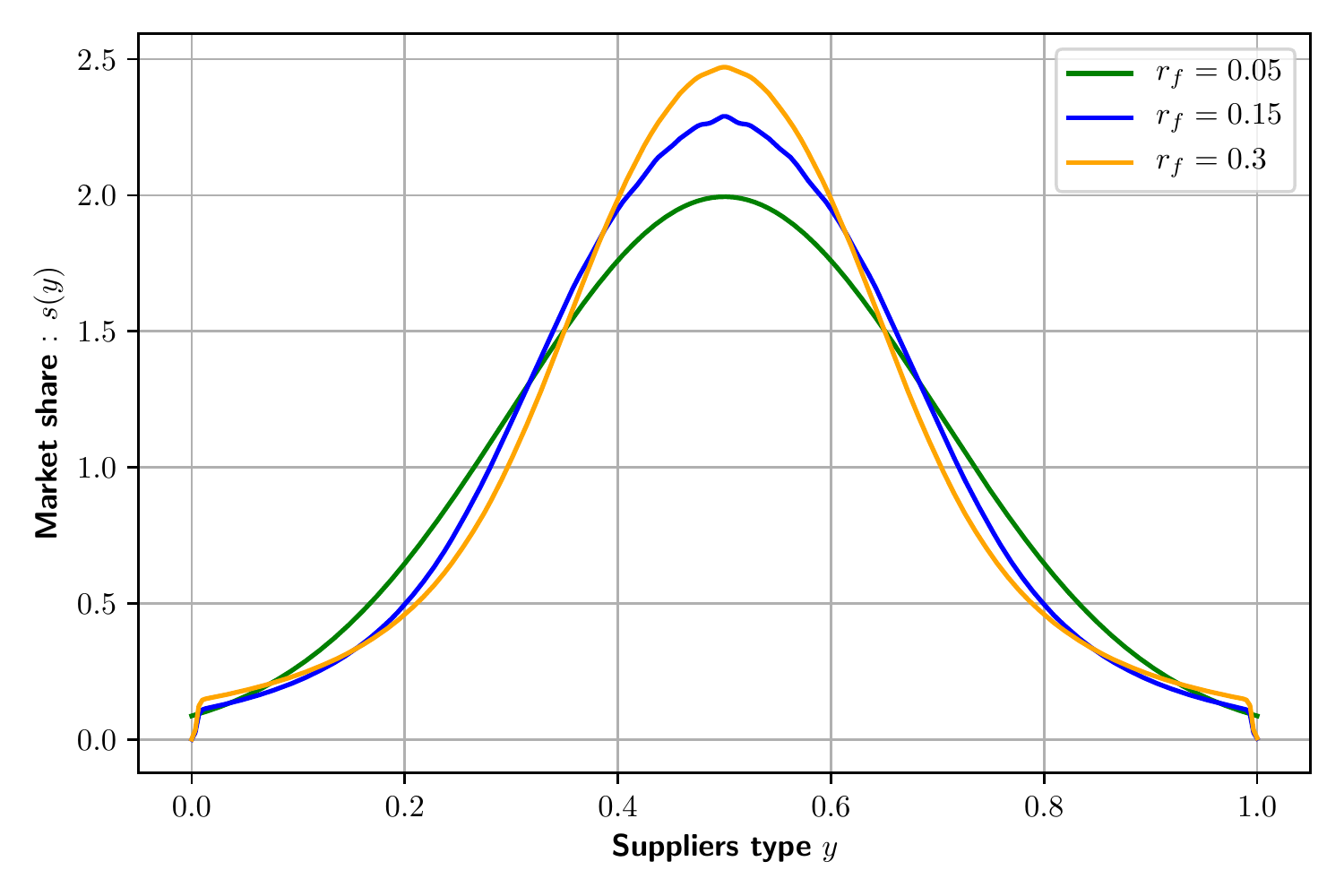}
        \caption{Strongly-heterogenous preferences, with normal-shaped preferences}
        \label{fig:syVSyPROP3d}
\end{subfigure}
\caption{Market share depending on firms type, for different intensity $r_f$ of search frictions, and with $\alpha = 0.75$}
\label{fig:finalsimulation}
\end{figure}

\end{document}